\documentclass[journal]{IEEEtran}
\usepackage{amssymb}
\usepackage{amsfonts}
\usepackage{amsmath}
\usepackage{algorithm}
\usepackage{graphicx,cite,multicol,multirow,diagbox,booktabs,array}
\usepackage{subfig}
\usepackage{verbatim}
\usepackage{cuted}
\usepackage{float}
\usepackage{bm}
\newtheorem{theorem}{Theorem}
\newtheorem{remark}{Remark}
\newtheorem{proof}{Proof}

\newtheorem{corollary}{Corollary}

\usepackage{etoolbox}
\pretocmd{\maketitle}{%
  \markboth{The definitive version was published in  in IEEE Transactions on Communications, vol. 70, no. 1, pp. 606-620, Jan. 2022, doi: 10.1109/TCOMM.2021.3120721.}{}%
}{}{}

\ifCLASSINFOpdf
\else
\fi

\begin{document}
\title{Joint Optimization for RIS-Assisted Wireless Communications: From Physical and Electromagnetic Perspectives}
\author{Xin Cheng,~Yan Lin,~Weiping Shi,~Jiayu Li,~Cunhua Pan,~Feng Shu, \\
Yongpeng Wu,~and~Jiangzhou~Wang,~\IEEEmembership{Fellow,~IEEE}
\thanks{Xin Cheng,~Yan Lin,~Weiping Shi,~and~Jiayu Li are with the School of Electronic and Optical Engineering, Nanjing University of Science and Technology, Nanjing, 210094, China. (e-mail: xincstar23@163.com). }
\thanks{Cunhua Pan is with the School of Electronic Engineering and Computer Science , Queen Mary University of London, Mile End Road
London E1 4NS, U.K.}
\thanks{Feng Shu is with the School of Information and Communication Engineering, Hainan University, Haikou 570228, China. and also with the School of Electronic and Optical Engineering, Nanjing University of Science and
Technology, Nanjing 210094, China.  }
\thanks{Yongpeng Wu is with the Shanghai Key Laboratory of Navigation and
Location Based Services, Shanghai Jiao Tong University, Minhang 200240,
China. }
\thanks{Jiangzhou Wang is with the School of Engineering and Digital Arts, University of Kent, Canterbury CT2 7NT, U.K.}}
\maketitle
\begin{abstract}
Reconfigurable intelligent  surfaces (RISs) are envisioned to be a disruptive wireless communication technique that is capable of reconfiguring the wireless propagation environment. In this paper,  we study a free-space RIS-assisted multiple-input single-output (MISO) communication system in far-field operation. To maximize the received power  from the physical and electromagnetic nature point of view, a comprehensive optimization, including beamforming of the transmitter, phase shifts of the RIS, orientation and position of the RIS is formulated and addressed.  After exploiting the property of line-of-sight (LoS) links, we derive closed-form solutions of beamforming and phase shifts.  For the non-trivial RIS position optimization problem in arbitrary three-dimensional space, a dimensional-reducing theory  is proved. The simulation results show that the proposed closed-form beamforming and phase shifts approach the upper bound of the received power. The robustness of our proposed solutions in terms of the perturbation is also verified. Moreover, the RIS significantly enhances the performance of the  mmWave/THz  communication system.
\end{abstract}


\IEEEpeerreviewmaketitle

\section{Introduction}
\IEEEPARstart{I}{n} recent years, wireless communication has witnessed great success in various aspects such as rate, stability and security. However, most of the existing techniques mainly rely on the transceiver design at both the transmitter and the receiver. The wireless propagation environment is left untouched. Unfortunately, the propagation loss and multi-path fading deteriorate the communication performance. Due to the rapid development of radio frequency (RF), micro electromechanical systems (MEMS) and metamaterial, a metasurface called reconfigurable  intelligent   surface (RIS) \cite{CuiCoding,HuanhuanA} has attracted a lot of attention. In \cite{LiaskosA}, a new wireless communication paradigm based on the concept of RIS was proposed, which can adaptively tune the propagation environment. The benefits and challenges were discussed in \cite{8910627}.

A RIS is composed of an array of  low-cost passive reflective elements, each of which can be controlled by a control loop to re-engineer the electromagnetic wave (EM) including steering towards any desired direction full absorption, polarization manipulation. The EM programmed by many reflective elements can be integrated constructively to induce remarkable effect. Unlike relay which requires active radio frequency (RF) chains, the RIS is passive because it dose not adopt any active transmit module (e.g., power amplifier)\cite{9119122}. Hence, the RIS is more energy efficient than the relay scheme.

Due to the above appealing features,  RIS-assisted communication systems have been studied extensively. For example, owing to  low-cost and passive reflecting elements,  RIS can achieve high spectrum and energy efficiency for future wireless networks \cite{9357969}. The joint design of beamforming and phase shifts was also investigated in various communication scenarios.  The contributions in \cite{DBLP:conf/globecom/WuZ18,DBLP:conf/iccchina/YuXS19,shi2019enhanced,8746155} showed that the RIS offers performance
improvement and coverage enhancement in  the single-user multiple-input single-output (MISO) system.  In downlink multi-user MISO case, the advantages of introducing RISs in enhancing the cell-edge user performance  were confirmed in \cite{Pan1}. Employing RISs to wireless information and power transfer (SWIPT) system in multi-user MIMO scenarios  was shown to beneficially enhance the system performance in terms of both the link quantity and the harvested power\cite{Pan2}. The advantages of introducing RIS were demonstrated in a secure multigroup multicast MISO communication system in \cite{9384498}. To minimize the symbol error rate (MSER) of an RIS-assisted  point-to-point multi-data-stream MIMO wireless communication system, the reflective elements at the RIS and the precoder at the transmitter were alternately optimized in \cite{9097454}.  With the assistance of RISs, secrecy communication rate i.e., physical layer security can be significantly improved \cite{8723525,8743496}. The work of \cite{DBLP:journals/corr/abs-2011-03726} examined the performance gain achieved by deploying an RIS in covert communications. RIS was proposed to create friendly multipaths for  directional modulation (DM) such that two confidential bit streams (CBSs) can be transmitted from Alice to Bob in \cite{DBLP:journals/corr/abs-2008-05067}. The important theoretical performance like ergodic spectral efficiency, symbol error probability, and outage probability was analysed and optimized in \cite{Pan1,8796421,8746155}. The channel estimation in the RIS-assisted scenario was studied in \cite{9370097,8879620}. More realistic limits  like discrete phase-shift and phase-dependent amplitude were considered in \cite{8930608}.  The authors in \cite{9133134} proposed a reflective modulation (RM) scheme for RIS-based communications, utilizing both the reflective patterns and transmit signals to carry information.

However, the above papers applied simple mathematical models that regard the  reflective elements of RIS as a diagonal matrix with phase shifts values, leading to relatively simplified algorithm designs and performance predictions.  There are some existing works that have studied the physical and electromagnetic models of the RIS in free space \cite{8936989,DBLP:journals/corr/abs-1912-06759,9206044,9184098}. The responses of  RISs to the radio waves were studied and based on this,  physical and electromagnetic path loss models in free space was established. The works showed that the path loss  model in near-field/far-field scenarios are of two kinds, depending on  the distance between the RIS and the transmitter/receiver. In the far-field condition, the spherical wave generated by the transmitter can be approximately regarded as a plane wave\cite{stutzman1997atd}. Thus it is more tractable for theoretical  analysis. The proposed free-space path loss model in \cite{9206044}, which is first validated through extensive simulation results, revealed the relationships between the free-space path loss of RIS-assisted wireless communications and the distances from the transmitter/receiver to the RIS, the size of the RIS, the near-field/far-field effects of the RIS, and the radiation patterns of antennas and unit cells. Moreover, the analytical model  matches quite well with the experiments in a microwave anechoic. Therefore, we apply the tractable and reliable path loss model proposed in \cite{9206044} to our work.

In this paper, we consider a far-field RIS-assisted MISO wireless communication system in free space.  The RIS-assisted free-space communication has been applied in many important scenarios, mainly in the field of UAV network\cite{9467275,9367288} and mmWave transmission\cite{9221316,9382000}, which have been both spotlighted as  next-generation communications.   To achieve the performance limit of this system from the physical and electromagnetic points of view, a comprehensive optimization, including  beamforming of the transmitter, phase shifts of the RIS and placement of the RIS, is formulated and addressed.

Our main contributions are summarized as follows:
\begin{enumerate}
\item The comprehensive optimization problem of far-field RIS-assisted wireless communication in free space, considering the physical and electromagnetic nature of RIS, is formulated. By exploiting the propagation property of electromagnetic wave in free space, under the principle of phase alignment and maximum ratio transmission, we derive closed-form solutions of beamforming and phase shifts to maximize the received power. Due to the extreme accuracy of the approximation technology, the performance of the closed-form solutions approaches the upper bound of the received power, verified in simulation part. Besides, the robustness of the proposed solutions in terms of the position perturbation of RIS is verified. It is also found that the optimal orientation of  RIS is just to perform specular reflection.  Moreover, a design principle of manufacturing the RIS is provided to countervail the deteriorating path loss of high-frequency electromagnetic wave.
\item  In order to reap full advantage brought by RIS, the problem of where to place an RIS in an arbitrary three-dimensional space is  discussed. By excavating the quasi-convex property of the objective function, we demonstrate that the optimal position is always on the boundary of the two-dimensional area of interest.  Then the theorem is naturally extended to arbitrary three-dimensional space.  This work provides the theoretical basis of   essentially reducing the dimension of the area of interesting. The simulation results show that the significant received power gain can be achieved,  owing to the assistance of the RIS at optimal position.
\end{enumerate}

The remainder of this paper is organized as follows. In Section II, we describe the system model of RIS-assisted wireless communication. After considering the physical and electromagnetic nature of the RIS, a joint optimization problem is formulated under far-field approximations. We propose the closed-form beamforming and phase shifts, as well as the optimal orientation of the RIS in Section III.  Based on above,  a reduced-dimension theory for finding the optimal position in the feasible space is proved in Section IV.  Moreover,  several important extensions of the aforementioned  works are discussed in Section V.  The proposed comprehensive scheme is numerically evaluated in Section VI. Finally, we draw conclusions in Section VII.

\emph{Notations:} Boldface lower case and upper case letters denote vectors and matrices, respectively. $(\cdot)^{H}$ denotes the conjugate and transpose operation and  $(\cdot)^{\ast}$ denotes the conjugate operation. $\mathbb{C}^{x\times y}$ denotes the space of $x\times y$ complex matrices.  $\mathbb{E}\{\cdot\}$ represents expectation operation. $\|\cdot\|$ denotes 2-norm. $\mathrm{diag}(\cdot)$ denotes a diagonal matrix whose diagonal elements are given by the corresponding vector. $\sigma_{max}(\cdot)$ represents the maximum singular value  of a matrix.  $\angle$ represents the element-wise taking angle  operation.


\section{System Model and Problem Formulation}

\subsection{System Model}
\begin{figure}
  \centering
  \includegraphics[width=0.5\textwidth]{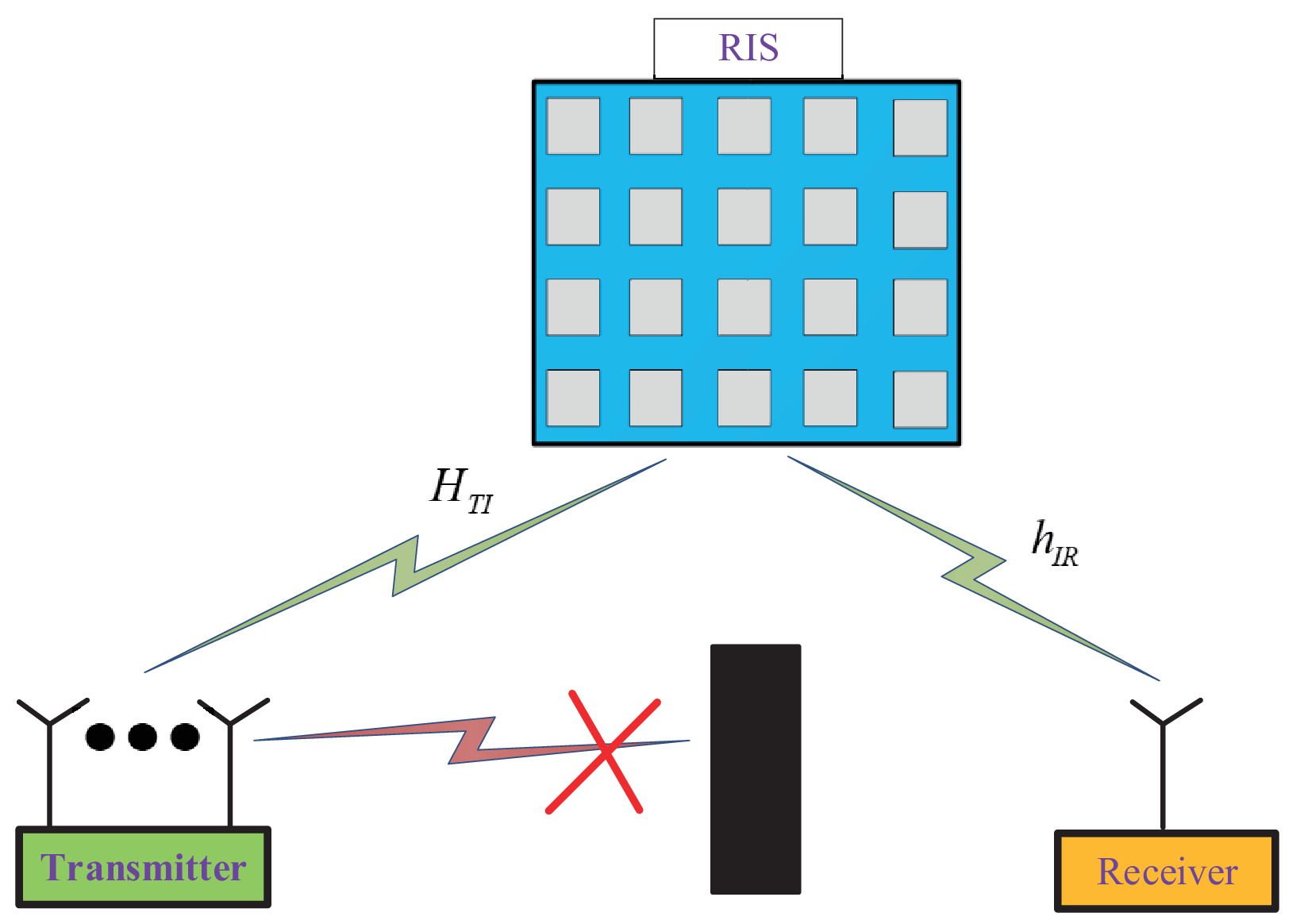}\\
  \caption{Diagram of RIS-assisted wireless communication system.}\label{Sys}
\end{figure}
We consider a multiple-input single-output (MISO) system where a RIS is deployed to assist in the wireless communication from the transmitter to the receiver, as illustrated in Fig.~\ref{Sys}. The transmitter is equipped with $N$ antennas, forming uniform linear array (ULA) while the user is equipped with a single antenna.  In this scenario, the direct link between the  transmitter and the receiver is blocked and the RIS is placed to provide a line-of-sight (LoS) link. The symmetric and regular RIS consists of $L$ reflective elements with $N_{I}$ rows and $M_{I}$ columns. The length of the single reflective element is denoted as $d_{x}$ and  the width is denoted as $d_{y}$.  For simple notations, the RIS, the transmitter, and the receiver are denoted as I, T, and R, respectively.

Let us define  ${\bf{H}}_{TI}^{H}\in \mathbb{C}^{L\times N}$ as the channel between the transmitter and the RIS and ${\bf{h}}_{IR}^{H}\in \mathbb{C}^{1\times L}$ as  the channel between the receiver and the RIS.   Using the amplitude and phase  to represent the plural element, the  channels are given by
\begin{subequations}\label{Channel1}
\begin{align}
\mathbf{H}_{TI}^H=\begin{bmatrix}
                      a_{TI,1,1}e^{j\theta_{TI,1,1}} &\cdots &  a_{TI,1,N}e^{j\theta_{TI,1,N}} \\
                                 \vdots      &\ddots &            \vdots \\
                     a_{TI,L,1}e^{j\theta_{TI,L,1}} &\cdots & a_{TI,L,N}e^{j\theta_{TI,L,N}} \\
                     \end{bmatrix},                                                                                                    \
\end{align}
\begin{align}
\mathbf{h}_{IR}^H=\left[a_{IR,1}e^{j\theta_{IR,1}}  ,\cdots, a_{IR,L}e^{j\theta_{IR,L}}\right].
\end{align}
\end{subequations}

Note that, in the LoS link, the channel element is mainly related to the position relationships of the communication parties. For example, $a_{TI,p,q}e^{j\theta_{TI,p,q}}$  is determined by the distance from the $p$-th antenna of T to the $q$-th reflective element directly. The establishment of the geometric model between communication parties is important. After selecting a coordinate system and choosing the center point of the transmitter/receiver and the RIS as reference points, we represent the positions of T, I, and R by vector $\mathbf{r}_{T}$, $\mathbf{r}_{I}$, and $\mathbf{r}_{R}$ respectively.  Let $\xi$ denote the  orientation of RIS. It is embodied in the elevation angle and the azimuth angle from T to I (denoted as  $\theta_{t}$ and  $\varphi_{t}$ respectively), and the elevation angle and the azimuth angle from R to the I (denoted as $\theta_{r}$ and  $\varphi_{r}$ respectively). Tuning the orientation of the RIS purposely means adjusting the RIS-related  angles ( $\theta_{t}$,  $\varphi_{t}$, $\theta_{r}$ and  $\varphi_{r}$) under geometric limitations, as illustrated in Fig.~\ref{notations}\subref{angle}. When representing the position of the $q$-th reflective element as $\mathbf{r}_{I,q}$, the centre point of the $q$-th reflective element is chosen as the reference point. In the same manner, $\mathbf{r}_{T,p}$ is used to denote the position of the $p$-th antenna of $T$.  Fig.~\ref{notations}\subref{geometricalre} depicts the positional relationship between the transmitter and the RIS. The geometrical relationship  between the receiver and RIS is a simplification of it, thus omitted here. The position relationship between the elements and the  RIS can be represented by a  function as $\mathbf{r}_{I,q}=f(\mathbf{r}_{I},\xi,q)$. The function $f$ is determined by  the shape and size of the RIS. Similarly, $\mathbf{r}_{T,p}=g(\mathbf{r}_{T},p)$, $g$ is a function determined by the direction of the ULA  and the antenna spacing $\Delta d_{T}$. Let $\Delta d_{T,p}$ denote the distance between the $p$-th antenna and the reference point of T, given by
\begin{align}\label{antennainterval}
\Delta d_{T,p}=(\frac{N+1}{2}-p)\Delta d_{T},~p=1\cdots N.
\end{align}
The distance from the $p$-th antenna element to the $q$-th RIS reflective element denoted by $d_{TI,p,q}$ can be expressed as
\begin{align}\label{d_TI}
d_{TI,p,q}=|\mathbf{r}_{I,q}-\mathbf{r}_{T,p}|,~p=1\cdots N,~q=1\cdots L.
\end{align}
It should be noticed that the rule of numbering sequence for antennas of T or elements of RIS is flexible and has no impact on the  results we concerned.

\begin{figure*}
    \subfloat[\label{angle}]{%
       \includegraphics[width=0.5\textwidth]{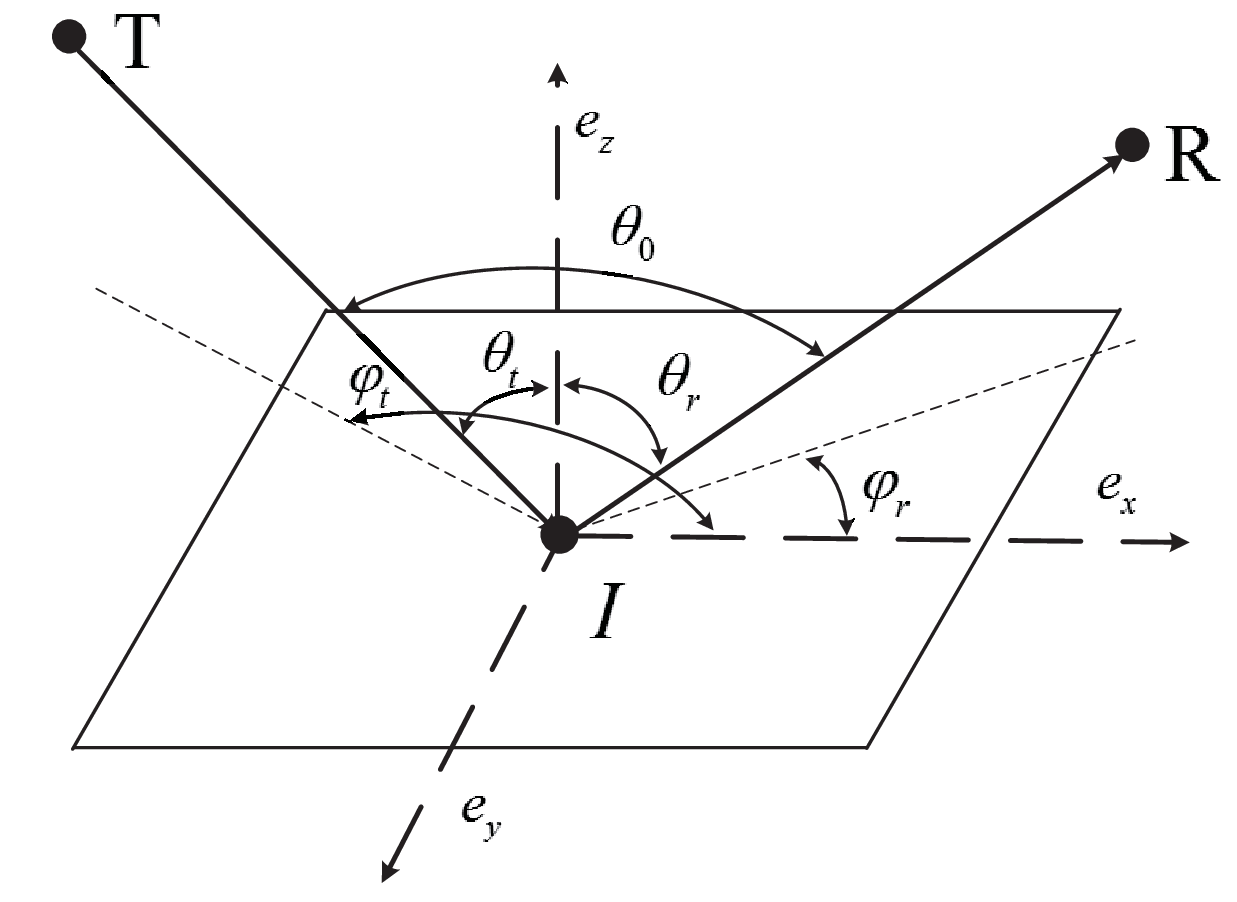}
   }
     \hfill
   \subfloat[\label{geometricalre}]{%
       \includegraphics[width=0.5\textwidth]{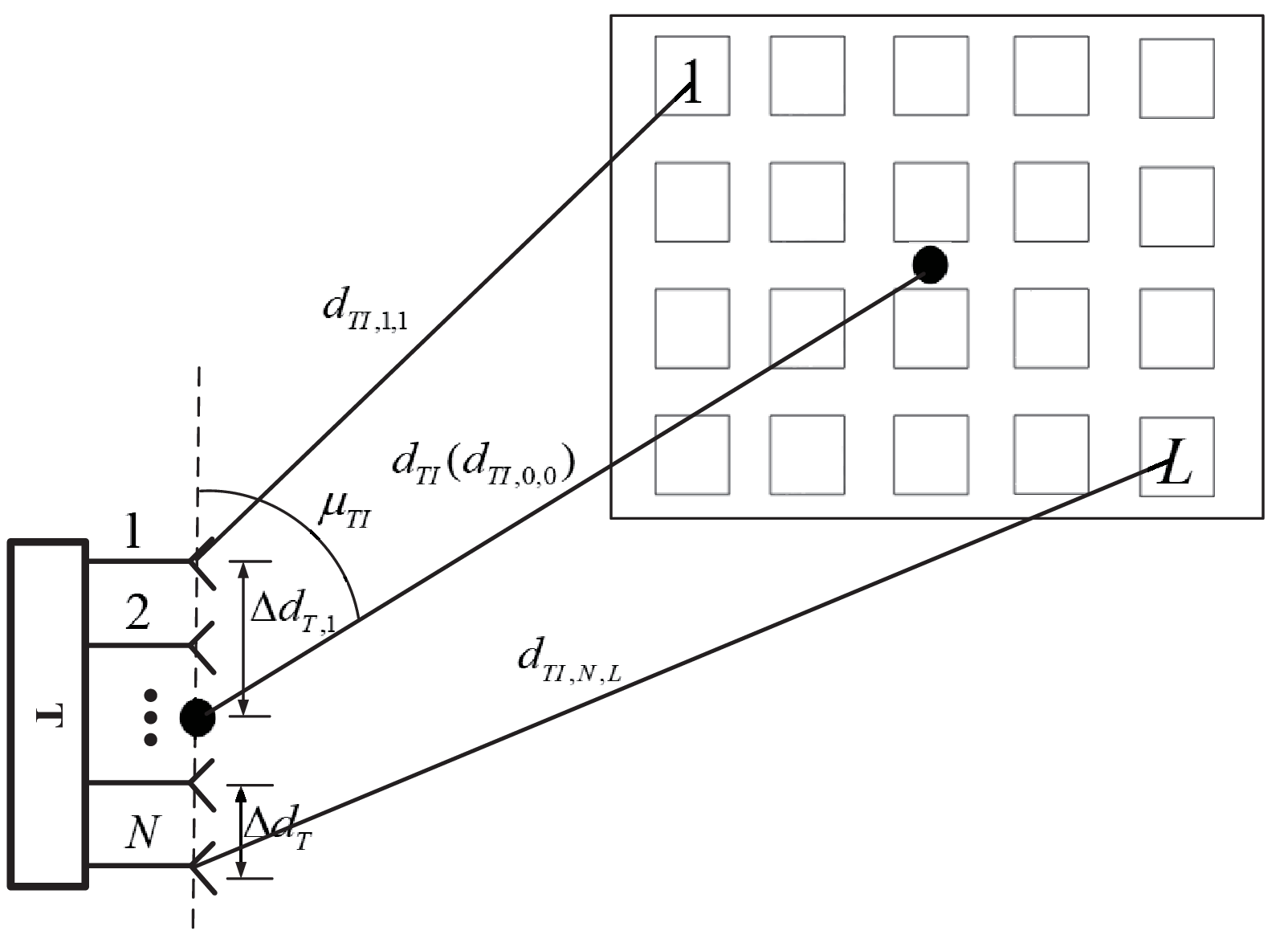}
     }
     \caption{Geometrical model of the RIS-assisted free-space wireless communication system. (a)  Angular relationships between the incident/reflected wave and the orientation of RIS.  (b) Positional relationships between the ULA of the transmitter and reflective elements of the RIS.}
     \label{notations}
\end{figure*}

In such a MISO system, the useful message $x$ is sent from  T to R, which is normalized such that $\mathbb{E}\left\{ x^{H}x\right\}=1$. Besides, they are multiplied by the beamforming vector $\mathbf{v}=(v_{1}, v_{2}, \cdots, v_{N})$  with the power limit $P_{t}$ ($\mathbf{v}^{H}\mathbf{v}\leq P_{t}$). In order to combine with the physical path loss model hereinafter, $P_{t}$ accounts for the signal power in the linear domain at unit distance.  The phase-shift matrix of the RIS, representing the properties of  the RIS, is denoted by $\bm{\Theta}=\mathrm{diag}(\bm{\theta})\in \mathbb{C}^{L\times L}$ with $\bm{\theta}=[\theta_{1}, \theta_{2}, \cdots, \theta_{L}]$.  The elements in $\bm{\Theta}$ satisfy the condition $\theta_{q}^{H}\theta_{q}=1,~q=1,\ldots L$. Equivalently,  $\theta_{q}=e^{j\varphi_{q}}$ with $\varphi_{q}$ is a real number.

The useful power of the  received signal can be expressed by
\begin{equation}\label{power}
P_{r}=|\mathbf{h}_{IR}^{H}\bm{\Theta}\mathbf{H}_{TI}^{H}\mathbf{v}|^{2}.
\end{equation}

\subsection{Physical and Electromagnetic Perspectives}

\begin{table*}
\centering
\caption{Notations about Physical and Electromagnetic Factors}
\begin{tabular}{|c|c|}
\hline
  $\mathrm{\mathbf{Symbol}}$ & $\mathrm{\mathbf{Definition}}$\\
  \hline
  $Z_{0}$  &    The characteristic impedance of the air\\
  \hline
  $G_{t}$  &    Antenna gain of  the transmitter\\
  \hline
  $G$      &    Gain of the RIS \\
  \hline
  $F$      &    Normalized power radiation pattern \\
  \hline
  $\mu_{TI}$ &  The complement angle of the DOA of the RIS at the transmitter \\
  \hline
  $\theta_{t}$ &  The elevation angle  from the reference   point of the RIS to the transmitter \\
             \hline
  $\varphi_{t}$ & The azimuth angle from the reference  point of the RIS to the transmitter\\
        \hline
  $\theta_{r}$ &  The elevation angle from the reference  point of the RIS to the receiver\\
  \hline
  $\varphi_{r}$ &  The azimuth angle from the  reference  point of the RIS to the receiver\\
  \hline
  $\Gamma$  &   The reflection coefficient of the RIS\\
  \hline
  $A_{t}$   & The aperture of the transmit antenna\\
  \hline
  $A_{r}$    &  The aperture of the receive antenna\\
  \hline
  $E^{r}$    &  The  electric field intensity  of the received signal\\
  \hline
\end{tabular}
\label{Para}
\end{table*}

In this subsection, we shall express the far-field amplitude gain and phase change of the RIS link from physical and electromagnetic perspectives.  The parameters mentioned in this subsection are explained in \textrm{Table \ref{Para}}.

In far-field operation, the length/width of the RIS is much smaller than the distances of communication parties. According to the propagation principle of electromagnetic wave, the amplitude gains from different antenna elements to different RIS elements can be assumed the same.
\begin{align}\label{farappeoximation}
&a_{TI,p,q}\approx a_{TI},~a_{IR,q}\approx a_{IR},~p=1\cdots N,~q=1\cdots L.
\end{align}
This kind of approximation is named  \emph{far-field amplitude approximation}. Since it's irrational to divide the effectiveness of RIS by $a_{TI}$ or $a_{IR}$, the joint term $a_{TI}a_{IR}\triangleq a_{TIR}$ is defined to represent the final amplitude gain from the transmitter to the receiver through individual reflective element of the RIS$\footnote{The amplitude gain is actually a combination of channel gain and antenna gain.}$.

Let us  concretize $a_{TIR}$.  Referring to the the electric fields model described in \cite{9206044}, we have
\begin{align}\label{am-gainTIR}
a_{TIR}&=\sqrt{\frac{|E^{r}|^2}{2Z_{0}}A_{r}} \nonumber\\
&=\sqrt{\frac{Z_{0}G_{t}Gd_{x}d_{y}F(\theta_{t},\varphi_{t})F(\theta_{r},\varphi_{r})\Gamma^{2}d_{TI}^{-2}d_{IR}^{-2}}{16\pi^2Z_{0}}\frac{G_{r}\lambda^2}{4\pi}} \nonumber\\
&=\sqrt{\frac{G_{t}G_{r}Gd_{x}d_{y}\lambda^2F(\theta_{t},\varphi_{t})F(\theta_{r},\varphi_{r})\Gamma^{2}}{64\pi^3}}d_{TI}^{-1}d_{IR}^{-1}   \nonumber\\
&\triangleq \delta_{TIR}d_{TI}^{-1}d_{IR}^{-1}.
\end{align}

Note that normalized power radiation pattern of the RIS is  denoted as $F(\theta_{t},\varphi_{t})F(\theta_{r},\varphi_{r})$. The general normalized power radiation pattern of a single reflective element is in the form of
\begin{subequations} \label{DeF}
\begin{align}
F(\theta_{t},\varphi_{t})=\left\{
\begin{aligned}
&\cos^{k} {\theta_{t}}  ~~ &\theta_{t}\in[0,\frac{\pi}{2}],\varphi_{t}\in[0,2\pi] \\
&0 ~~ &\theta_{t}\in(\frac{\pi}{2},\pi],\varphi_{t}\in[0,2\pi],
\end{aligned}
\right.  \\
F(\theta_{r},\varphi_{r})=\left\{
\begin{aligned}
&\cos^{k} {\theta_{r}}  ~~ &\theta_{r}\in[0,\frac{\pi}{2}],\varphi_{r}\in[0,2\pi] \\
&0 ~~ &\theta_{r}\in(\frac{\pi}{2},\pi],\varphi_{r}\in[0,2\pi],
\end{aligned}
\right.
\end{align}
\end{subequations}
where $k\geq0$$\footnote{The form $\cos^{k}$ can be used to match the normalized power radiation pattern of different unit cell and antenna designs with an appropriate $k$ \cite{stutzman1997atd}.}$. Thus, turning a RIS  at a fixed position impacts the amplitude gain $a_{TIR}$. Fig.~\ref{notations}\subref{angle} illustrates this problem. From above, the amplitude gain of the RIS is essentially a function of the position  and orientation of the RIS.

We now turn the attention to the phase changes of the EM in the RIS link. The phase changes in the channel $\mathbf{H}_{TI}^{H}$   are equivalently  written as
\begin{subequations}
\begin{equation}
\theta_{TI,q,p}=2\pi\frac{d_{TI,p,q}}{\lambda}\triangleq 2\pi\frac{d_{TI}+\Delta d_{TI,p,q}}{\lambda},
\end{equation}
\begin{equation}
\theta_{IR,q}=2\pi\frac{d_{IR,p}}{\lambda}\triangleq 2\pi\frac{d_{IR}+\Delta d_{IR,q}}{\lambda},
\end{equation}
\end{subequations}
where $\lambda$ is the carrier wavelength.  Because the size of the reflective element and the antenna element separation is the same order or sub-order of the carrier wavelength, unlike amplitude gains, different phase changes can't be assumed to be the same. However, in  far-field operation,  a tight distance approximation is suitable to apply, which is shown at the top of next page.
\newcounter{mytempeqncnt}
\begin{figure*}[ht]
\normalsize
\setcounter{mytempeqncnt}{\value{equation}}
\setcounter{equation}{8}
\begin{subequations}\label{farappro}
\begin{align}
\Delta d_{TI,p,q} &=d_{TI,p,q}-d_{TI,p,0}+d_{TI,p,0}-d_{TI} \nonumber\\
&\overset{a}{\approx}d_{TI,p,q}-d_{TI,p,0}+\Delta d_{T,p}\cos{\mu_{TI}}  \nonumber\\
&\overset{b}{\approx}-\sin\theta_{t,p}\cos\varphi_{t,p}(m_{q}-\frac{M_{I}+1}{2})d_{x}-\sin\theta_{t,p}\sin\varphi_{t,p}(n_{q}-\frac{N_{I}+1}{2})d_{y}
+(\frac{N+1}{2}-p)\Delta d_{T}\cos{\mu_{TI}}  \nonumber\\
&\overset{c}{\approx}\underbrace{-\sin\theta_{t}\cos\varphi_{t}(m_{q}-\frac{M_{I}+1}{2})d_{x}-\sin\theta_{t}\sin\varphi_{t}(n_{q}-\frac{N_{I}+1}{2})d_{y}}_{\Delta d^{T}_{I,q}}+\underbrace{(\frac{N+1}{2}-p)\Delta d_{T}\cos{\mu_{TI}}}_{\Delta d^{I}_{T,p}},
\end{align}
\begin{align}
\Delta d_{IR,q}&=d_{IR,q}-d_{IR}\approx \underbrace{-\sin\theta_{r}\cos\varphi_{r}(m_{q}-\frac{M_{I}+1}{2})d_{x}-\sin\theta_{r}\sin\varphi_{r}(n_{q}-\frac{N_{I}+1}{2})d_{y}}_{\Delta d^{R}_{I,q}}.
\end{align}
\end{subequations}
\setcounter{equation}{9}
\hrulefill
\vspace*{4pt}
\end{figure*}
In the formula (\ref{farappro}), $\theta_{t,p}/\varphi_{t,p}$ represents the elevation/azimuth angle at the RIS from the $p$-th transmitting antenna to the RIS.  $m_{q}$ is the index number of columns of the $q$-th element and $n_{q}$ is the index number of rows of the $q$-th element.  It needs to be mentioned  that there are $\sin(\cdot)$ functions instead of $\cos(\cdot)$ functions in existing works\cite{6831645,5996700} due to the use of the supplementary angle. This kind of approximation is named  \emph{far-field phase approximation}.

In more details, approximation $(a)$ requires $d_{TI} \gg \Delta d_{T,p}$, similar to the traditional MIMO model\cite{DBLP:books/cu/G2005}. Approximation $(b)$ is based on the formula $eq.~32$ in \cite{9206044}, which requires $d_{TI}\gg \sqrt{m_{q}^2d_{x}^2+n_{q}^2d_{y}^2}$. We assume that the elevation/azimuth angles at the RIS from different antennas of the transmitter are identical, deriving approximation $(c)$. This assumption is reasonable when  $d_{TI} \gg \Delta d_{T,p}$. As a conclusion, The  conditions of \emph{far-field approximation}are summarized as follows
\begin{align}
d_{TI}\gg N\Delta d_{T}, d_{TI}\gg L \sqrt{d_{x}^2+d_{y}^2}, d_{IR} \gg L \sqrt{d_{x}^2+d_{y}^2}.
\end{align}

\subsection{Problem Formulation}
From above, we can see that the channel gain and phase changes  are both  tightly related to the position and the orientation of the RIS. Some existing works, considering a fixed RIS, have  designed  the beamforming vector $\mathbf{v}$ and phase-shift matrix $\bm{\Theta}$  jointly to improve the received power, namely the information  achievable rate. In this paper, the position of the RIS $\mathbf{r}_{I}$ and the orientation of the RIS $\xi$ are also  under-determined variables, thus they can be utilized to further enhance the received power. Considering practical limitations, the joint optimization problem of RIS-assisted  wireless communication is given by
\begin{align}\label{P1}
\mathrm{(P1)}:&\max_{\mathbf{v},\bm{\Theta},\mathbf{r}_{I},\xi}~~~~P_{r}\nonumber\\
&~~\text{s. t.}~~\mathbf{v}^{H}\mathbf{v}\leq P_{t}\nonumber\\
&~~~~~~~~~\bm{\Theta}\in\mathbb{B}\nonumber\\
&~~~~~~~~~\{\mathbf{r}_{I},\xi\}\in \mathbb{S}_{0} \nonumber\\
&~~~~~~~~~\mathbf{r}_{I}\in \mathbb{S}_{1}\cap \mathbb{S}_{2},
\end{align}
where $\mathbb{S}_{0}$ is a set that guarantees the EM to propagate from T to R through the RIS directly, which is specified by application environment. In more details, it is determined by the relative positional relationship of communication parties, obstacles and the orientation of RIS. $\mathbb{S}_{1}$ can be expressed as $(d_{TI}\geq r_{0})\cap (d_{IR}\geq r_{1})$. $r_{0}$ and $r_{1}$ are the minimum distance to guarantee the far-field condition, respectively. $\mathbb{S}_{2}$ is the feasible area that the RIS can be fixed at. Herein, we denote the feasible set of $\bm{\Theta}$ as $\mathbb{B}$ here. Without loss of generality, we assume that the antennas in T and R are omni-directional$\footnote{If we consider the directional antenna, only placing the RIS in the main lobe is meaningful. Let $\mathbb{S}_{3}$ denotes the area in the main lobe.  Therefore, a new restriction $\mathbf{r}_{I}\in\mathbb{S}_{3}$ should be added to $\mathrm{(P1)}$.}$.

The problem ($\mathrm{P1}$) is a joint optimization of four variables, and it is challenging to solve directly. In the following, we devote to solving it by two phases.  In the first phase, we assume  $\mathbf{r}_{I}$  is fixed, and the closed-form global optimal solutions of $\bm{\Theta}$ and $\mathbf{v}$ as well as  the optimal $\xi$  are proposed. In the second phase, based on the optimal solutions of the other three variables for a fixed $\mathbf{r}_{I}$, the problem $\mathrm{(P1)}$ become an unadulterated position optimizing problem ($\mathrm{P2}$).  Substituting the optimal position of ($\mathrm{P2}$) back to the optimal $\bm{\Theta}$ and $\mathbf{v}$ and $\xi$, which can be treated as  functions of $\mathbf{r}_{I}$, the global optimal solutions of $\mathrm{(P1)}$ are all concretized$\footnote{However, since we focus on far-field communications, the channels have been already approximated according to the far-field approximations in ($\mathrm{P1}$). Therefore, the global optimal solutions of  $\mathrm{(P1)}$ are near-optimal to the primal  problem without approximations.}$. Note that, the sub-optimization problems in each phase can also be viewed as meaningful independent works.

\section{Optimal Strategy with a Fixed RIS}
To decouple the solutions of $\bm{\Theta}$, $\mathbf{v}$, $\xi$  with $\mathbf{r}_{I}$ in problem $\mathrm{(P1)}$, we consider a fixed RIS in this section. We find that the  $\bm{\Theta}$ and $\mathbf{v}$ are coupled with each other, but irrelevant to $\xi$. In addition to derive optimal solutions of $\bm{\Theta}$, $\mathbf{v}$ and $\xi$, we propose an
\emph{anti-decay designing  principle} of manufacturing RIS for free-space THz communication.

\subsection{Optimal $\bm{\Theta}$ and $\mathbf{v}$}
The solutions of $\bm{\Theta}$ and $\mathbf{v}$ for the fixed RIS can be obtained using the iterative method proposed in \cite{DBLP:conf/globecom/WuZ18}. However, it's challenging to solve the joint optimization problem $\mathrm{(P1)}$ based on it. Herein, by exploiting the tightly coupled property of channel elements in the free space, we provide the analytic and optimal solutions of $\bm{\Theta}$ and $\mathbf{v}$.

For any given $\mathbf{\Theta}$, it is widely known that maximum-ratio transmission (MRT) is the optimal transmit beamforming to maximize the received power  \cite{DBLP:conf/globecom/WuZ18}, i.e.,
\begin{align}
\mathbf{v}^{\star}=\sqrt{P_{t}}\frac{(\mathbf{h}_{IR}^{H}\bm{\Theta}\mathbf{H}_{TI}^{H})^{\ast}}{\left\|\mathbf{h}_{IR}^{H}\bm{\Theta}\mathbf{H}_{TI}^{H}\right\|}.
\end{align}
Applying the optimal $\mathbf{v}^{\star}$, the received power can be expressed by
\begin{align}\label{Pv}
P_{r}=\left\|\mathbf{h}_{IR}^{H}\bm{\Theta}\mathbf{H}_{TI}^{H}\right\|^2P_{t}.
\end{align}

According to the \emph{far-field gain  approximation}  and  \emph{far-field phase approximation}, the channel matrix $\mathbf{H}_{TI}^{H}$ can be rewritten as
\begin{subequations}\label{deb}
\begin{align}
\mathbf{H}_{TI}^{H}=a_{TI}e^{j\frac{2\pi}{\lambda}d_{TI}}\mathbf{a}\mathbf{b}^{T},
\end{align}
\begin{align}
\mathbf{a}=\left[e^{j\frac{2\pi}{\lambda}\Delta d^{T}_{I,1}}, e^{j\frac{2\pi}{\lambda}\Delta d^{T}_{I,2}},\cdots, e^{j\frac{2\pi}{\lambda}\Delta d^{T}_{I,L}},\right],
\end{align}
\begin{align}
\mathbf{b}=\left[e^{j\frac{2\pi}{\lambda}\Delta d^{I}_{T,1}}, e^{j\frac{2\pi}{\lambda}\Delta d^{I}_{T,2}},\cdots, e^{j\frac{2\pi}{\lambda}\Delta d^{I}_{T,N}},\right].
\end{align}
\end{subequations}
To maintain the uniform, we rewrite the $\mathbf{h}_{IR}^{H}$ as
\begin{subequations}
\begin{align}
\mathbf{h}_{IR}^{H}=a_{IR}e^{j\frac{2\pi}{\lambda}d_{IR}}\mathbf{c}^{T},
\end{align}
\begin{align}
\mathbf{c}=\left[e^{j\frac{2\pi}{\lambda}\Delta d^{R}_{I,1}}, e^{j\frac{2\pi}{\lambda}\Delta d^{R}_{I,2}},\cdots, e^{j\frac{2\pi}{\lambda}\Delta d^{R}_{I,L}},\right].
\end{align}
\end{subequations}

After substituting them into (\ref{Pv}), the new form of the received power is given by
\begin{align}\label{powerform}
P_{r}&=a_{TI}a_{IR}\left\|\mathbf{c}^{T}\bm{\Theta}\mathbf{a}\mathbf{b}^{T}\right\|^2P_{t}\nonumber\\
&=a_{TIR}\left\|\bm{\theta}^{T}\mathrm{diag}(\mathbf{c}^{T})\mathbf{a}\mathbf{b}^{T}\right\|^2P_{t}\nonumber\\
&=a_{TIR}\left\|\bm{\theta}^{T}\mathbf{d}\mathbf{b}^{T}\right\|^2P_{t}=a_{TIR}N(\bm{\theta}^{T}\mathbf{d})^2P_{t},
\end{align}
with
\begin{align}\label{ded}
\mathbf{d}=\left[e^{j\frac{2\pi}{\lambda}\Delta d^{T}_{I,1}+\Delta d^{R}_{I,1}}, \cdots, e^{j\frac{2\pi}{\lambda}\Delta d^{T}_{I,L}+\Delta d^{R}_{I,L}},\right].
\end{align}
Hence, the optimal $\bm{\theta}$ to maximize the received power is
\begin{align}
\bm{\theta}^{\star}=\mathbf{d}^{\ast}
\end{align}
More clearly, considering the expressions of $\Delta d$ terms in (\ref{farappro}), the optimal phase shift of reflective element $q$ is designed as
\begin{subequations}\label{opphase}
\begin{align}\label{opphasemain}
\varphi_{q}^{\star}\triangleq &2\pi\frac{1}{\lambda}\big((\sin\theta_{t}\cos\varphi_{t}+\sin\theta_{r}\cos\varphi_{r})(m_{q}-\frac{M_{I}+1}{2})d_{x} \nonumber\\  &+(\sin\theta_{t}\sin\varphi_{t}+\sin\theta_{r}\sin\varphi_{r})(n_{q}-\frac{N_{I}+1}{2})d_{y}\big),
\end{align}
\begin{align}
\theta_{q}^{\star}=e^{j\varphi_{q}^{\star}}.
\end{align}
\end{subequations}

With the closed-form $\bm{\theta}^{\star}$, the decoupled analytic solution of $\mathbf{v}^{\star}$  is given by
\begin{align}\label{opbeaming}
\mathbf{v}^{\star}=\frac{\sqrt{P_{t}}}{\sqrt{N}}\mathbf{b}^{\ast}.
\end{align}
More clearly, considering the expressions of $\Delta d$ terms in (\ref{farappro}), the $p$-th element of $\mathbf{v}$  is designed as
\begin{align}
v_p^{\star}=\frac{\sqrt{P_{t}}}{\sqrt{N}}e^{j-\frac{2\pi}{\lambda}(\frac{N+1}{2}-p)\Delta d_{T}\cos{\mu_{TI}}}.
\end{align}

Under the proposed closed-form beamforming and phase shifts, the optimal received power becomes
\begin{align}\label{postiontopower1}
P_{r}=NL^{2}a_{TIR}^2P_{t}.
\end{align}
As seen, the growth of the received power follows a $N$ scaling with the antenna elements and $L^{2}$ scaling with the reflective elements, which has been recognized
in several recent works \cite{9184098,8910627}.

The optimal received power can be expended  as
\begin{align}\label{powerwavelength}
P_{r}&=(L\lambda)^2\frac{G_{t}G_{r}Gd_{x}d_{y}F(\theta_{t},\varphi_{t})F(\theta_{r},\varphi_{r})\Gamma^{2}}{64\pi^3}d_{TI}^{-2}d_{IR}^{-2}NP_{t}\nonumber\\
&=\frac{\lambda^2}{d_{x}d_{y}}\frac{S^2_{RIS}G_{t}G_{r}GF(\theta_{t},\varphi_{t})F(\theta_{r},\varphi_{r})\Gamma^{2}}{64\pi^3}d_{TI}^{-2}d_{IR}^{-2}NP_{t},
\end{align}
where $S_{IRS}$ denotes the total area of the RIS. The formula (\ref{powerwavelength}) indicates an  \emph{anti-decay designing  principle} of manufacturing RIS for  mmWave/THz communication in free space. To maintain the the received power as the
frequency of carrier wave increases, the RIS must manufactured according to the wavelength. In more details, there are two alternatives: One is when the  area of total RIS is limited, the size of   the reflective element should be  in a  fix proportional to the carrier wavelength. The other is when the area of reflective element is fixed, the number of the reflective elements is in  a fixed inverse ratio to the  carrier wavelength.

\subsection{Optimal Orientation of the RIS}
As mentioned before, the orientation of the RIS  can  be adjusted to improve the received power via maximizing $F(\theta_{t},\varphi_{t})F(\theta_{r},\varphi_{r})$. The optimal value is denoted as $F^{\star}$ hereinafter.  Due to large distances between the RIS and communication parties in the far-field case, it's reasonable to treat all reflective elements as  the reference point of the RIS when calculating $F^{\star}$.  As illustrated in Fig.~\ref{notations}\subref{angle}, the problem of turning a RIS to reap the max gain  is organized as
\begin{align}\label{problemturning}
&\max_{\theta_{t},\theta_{r}}~~F(\theta_{t},\varphi_{t})F(\theta_{r},\varphi_{r})\nonumber\\
&~~\text{s. t.}~~\theta_{t}\in (0,\frac{\pi}{2})\nonumber\\
&~~~~~~~~\theta_{r}\in (0,\frac{\pi}{2})\nonumber\\
&~~~~~~~~\theta_{0}\leq\theta_{t}+\theta_{r}\leq 2\pi-\theta_{0}.
\end{align}
It is found that when $\theta_{t}=\theta_{r}=\frac{\theta_{0}}{2}$, the objective function in (\ref{problemturning}) reaches the max value, denoted as $F^{\star}$, which is given by
\begin{align}\label{OpFF}
F^{\star}&=(\cos^{2}{\frac{\theta_{0}}{2}})^{k}=(\frac{1}{2}\cos{\theta_{0}}+\frac{1}{2})^{k} \nonumber\\
&=(\frac{d_{TI}^2+d_{IR}^2-d_{TR}^2}{4d_{TI}d_{IR}}+\frac{1}{2})^{k}.
\end{align}
From the above, the optimal orientation $\xi$  is just to make the RIS perform specular reflection, that is, the incident signal is mainly reflected towards the mirror direction $(\theta_{r} =\theta_{t})$. It is worth mentioned that this result in the MISO system is consistent with the counterpart in the SISO system\cite{9206044}.

\section{Optimal position of the RIS}
In this section, jointly with the optimal solutions of $\bm{\Theta}$, $\mathbf{v}$, $\xi$, we aim to study the optimal position of the RIS in problem $(\mathrm{P1})$.

As the objective function to optimize the position of the RIS, the maximum received power (\ref{powerwavelength}) is simplified as
\begin{align}\label{objectfunc}
P_{r}&=NL^2 \delta_{TIR}^2d_{TI}^{-2}d_{IR}^{-2}=NL^2\frac{\delta_{TIR}^2}{F^{\star}}F^{\star}d_{TI}^{-2}d_{IR}^{-2}\nonumber\\
&=NL^2\frac{\delta_{TIR}^2}{F^{\star}}\underbrace{(\frac{d_{TI}^2+d_{IR}^2-d_{TR}^2}{4d_{TI}d_{IR}}+\frac{1}{2})^{k}d_{TI}^{-2}d_{IR}^{-2}}_{F_{object}}.
\end{align}
According to (\ref{am-gainTIR}), $\frac{\delta_{TIR}^2}{F^{\star}}$ is a constant with no relationship to the position of the RIS. The position change of the RIS only impacts the value of  $d_{TI}$ and  $d_{IR}$ in this formula.

The position optimizing question is formulated as
\begin{align}\label{P2}
\mathrm{(P2):}&\max_{\mathbf{r}_{I}}~~ F_{object}  \nonumber\\
&~~\text{s. t.} ~~\mathbf{r}_{I}\in\mathbb{S},
\end{align}
where $\mathbb{S}=\mathbb{S}_{0}\cap\mathbb{S}_{1}\cap \mathbb{S}_{2}$  represents the feasible space to place RIS. In this section, we'll study the  property of optimal position in a two-dimensional plane, then extend it to arbitrarily three-dimensional space.

It is known that a three-dimensional feasible space can be fully split  into parallel two-dimensional subspaces. Then, the problem $\mathrm{(P2)}$ is replaced by many parallel subproblems, where the position of the RIS is an two-dimensional variable. It is noted that the principle of splitting is flexible and has no influence to the final result.

We start from  not considering the constraints of  placing RIS,. Then the parallel subproblems is to optimize the position of RIS on an infinite large plane, named S. Adequately, we consider the $\frac{1}{2}$ plane of S due to the symmetry. Hereinafter, the plane S represents the half plane. The scenario of finding optimal position on plane S is illustrated in  Fig.~\ref{planeS}. The points $T'$ and $R'$ represent the projection of T and R on plane S respectively.  The line l denotes the line containing $T'$ and $R'$ on plane S. The term $l_{R'T'}$ denotes the  line segment whose  endpoints are $T^{'}$ and $R^{'}$. The term $l_{\overrightarrow{T'R'}}$ accounts for  the  half-line staring from $T^{'}$  and containing $R^{'}$. By contrary,  the  half-line  $l_{\overrightarrow{T'R'}}$ stars from $R^{'}$  and contains $T^{'}$. The term $\nu_{TI}$ is the intersection angle of $l_{TR}$ and $l_{TT'}$.
\begin{figure}
  \centering
  \includegraphics[width=0.5\textwidth]{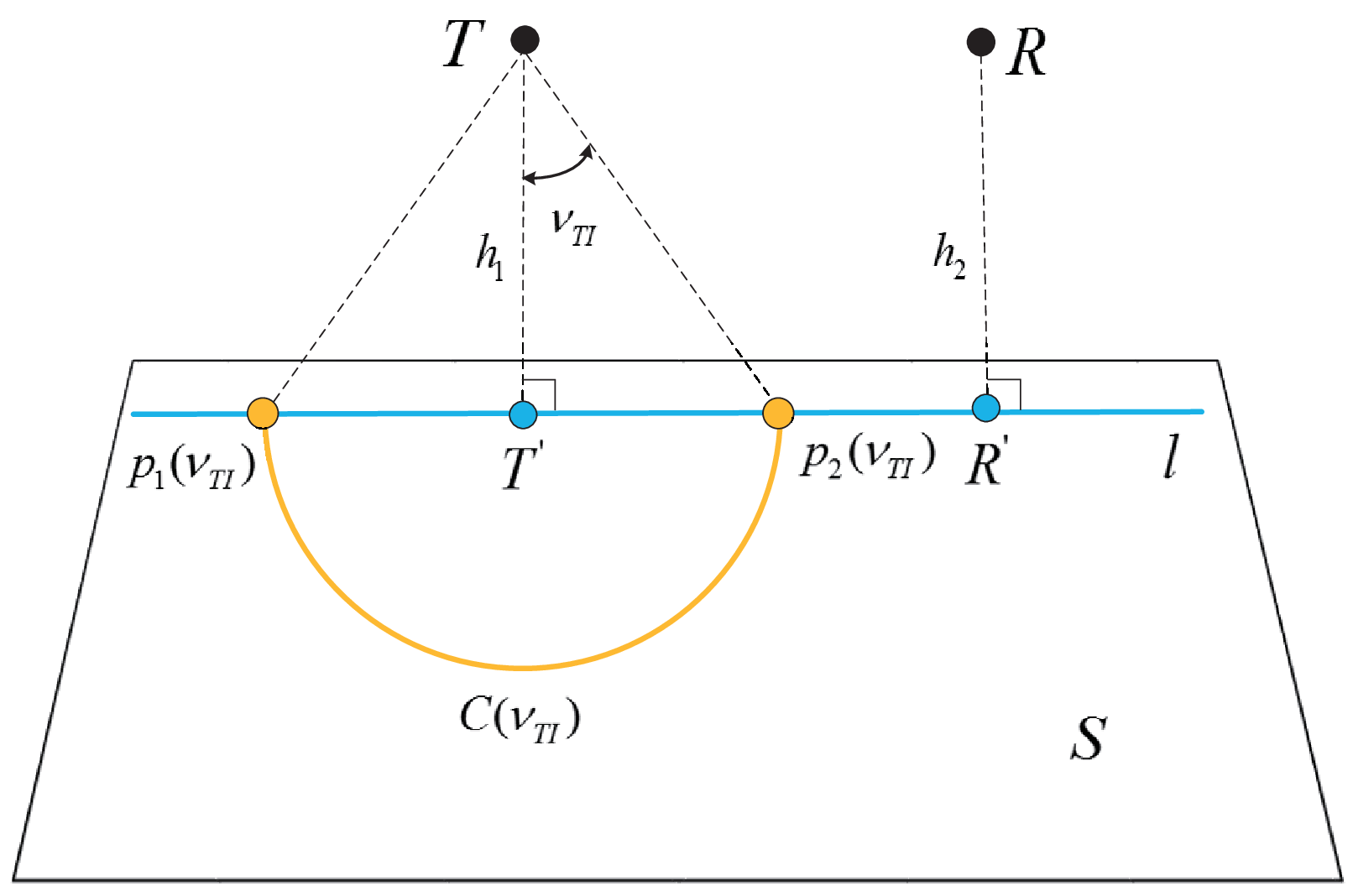}\\
  \caption{Diagram of placing the reference point of the RIS on a plane S.}\label{planeS}
\end{figure}

\begin{figure}
  \centering
  \includegraphics[width=0.5\textwidth]{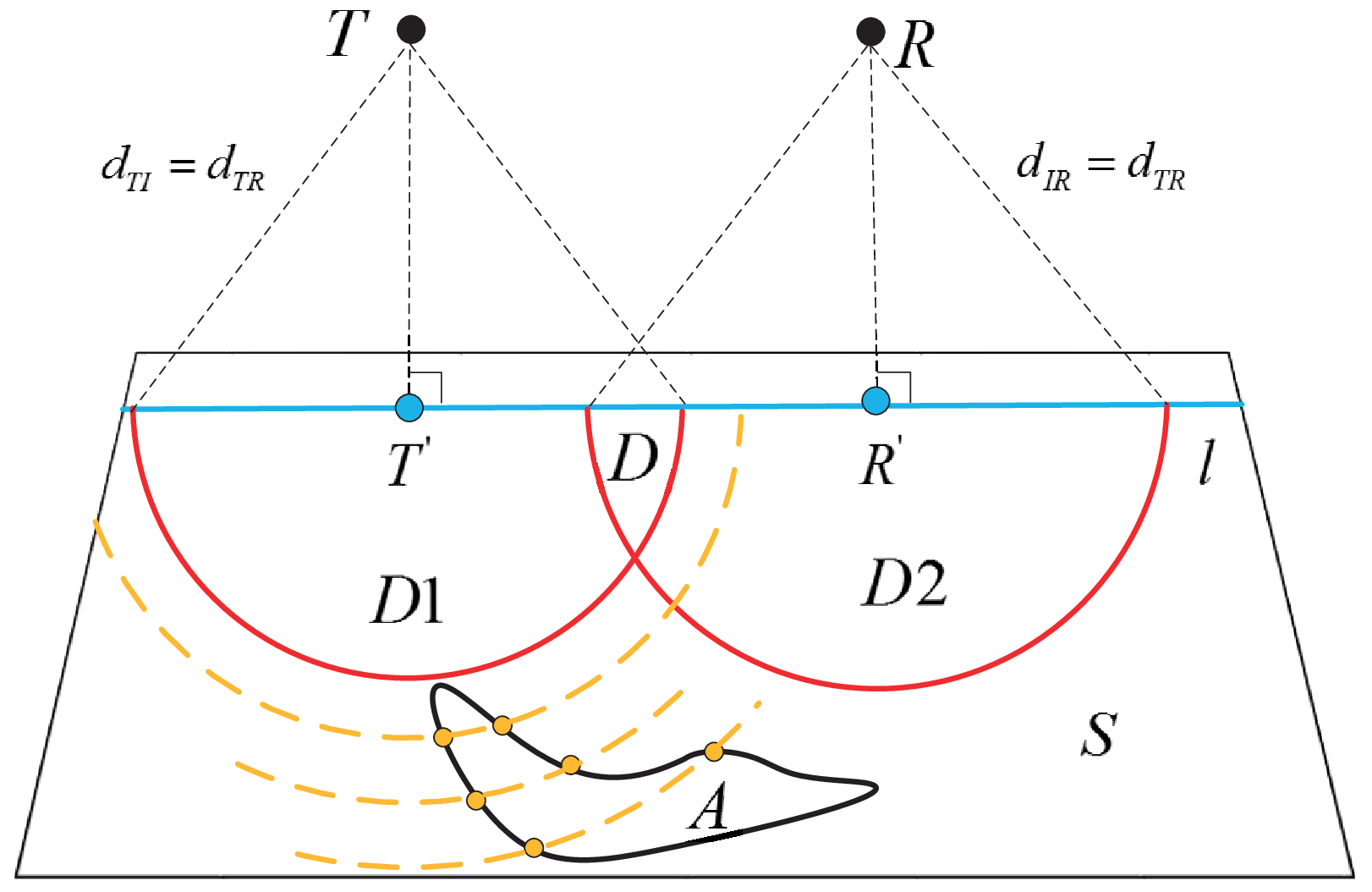}\\
  \caption{Illustration of area D and corollary 1. The yellow point denotes the possible optimal position of the RIS on the intersection of A and $C(\nu_{TI})$.}\label{sideproof}
\end{figure}

It is obvious that the values of $d_{IR}$ and $\nu_{TI}$ can uniquely determine a position of the RIS on plane S.  The optimal problem $\mathrm{(P2)}$ is reformulated as:
\begin{align}
&\max_{d_{IR},\nu_{TI}}~~~~F_{object}        \nonumber\\
&~~\text{s. t.}~~(d_{IR},\nu_{TI})\in \mathbb{C}     \nonumber\\
&~~~~~~~~~d_{TI}=\frac{h_{1}}{\cos{\nu_{TI}}},
\end{align}
where, $\mathbb{C}$ represents  possible unions of $d_{IR}$ and $\nu_{TI}$.  For a given $\nu_{TI}$, the feasible set of the RIS on plane S is a circle marked as $C_{\nu_{TI}}$. After excavating the quasi-convex property of $F_{object}$, the following theorems are obtained.

\begin{theorem}\label{th1}
If $k=0$ in (\ref{DeF}), on plane S, the optimal position of the RIS is must on the line segment $l_{R'T'}$. If $k>0$ in (\ref{DeF}), on the area S-D, the optimal position is must on  line l.The area D is a  set of  $(d_{TI}\leq d_{TR})\cap (d_{TI}\leq d_{TR})$,  as shown in Fig.~\ref{sideproof}.
\end{theorem}
\begin{proof}
See Appendix A and Appendix B.
\end{proof}

As illustrated in Fig.~\ref{sideproof}, we put forward a useful corollary  for placing the RIS on an arbitrary two-dimensional area A.
\begin{corollary}
For an arbitrary closed feasible area A on the whole plane S when $k=0$ or S$-$D when $k>0$, the optimal position of the RIS is must on A's boundary and the feasible part on line l.
\end{corollary}

Eventually, based on the fact that a three-dimensional space can be split to the parallel planes fully, we extend above corollaries to the feasible space $\mathbb{S}$. Similar to the two-dimensional case, a special space is defined as the forbidden space, whose cross section created by plane S are the area D.  The following corollary is derived.
\begin{corollary}
For an arbitrary closed feasible space $\mathbb{S}$ in the whole three-dimensional space when $k=0$ or in  the three-dimensional space except the special space when $k>0$, the optimal position of the RIS is must on the surface of $\mathbb{S}$ $\footnote{Note that, due to the plane S is an arbitrary plane,  the feasible part on line l can be avoided by selecting the plane S suitablely.}$.
\end{corollary}

\begin{remark}
Note that the area D diminishes as the $h_{1}$ and $h_{2}$ becomes either larger or smaller. For most cases, the area D can be neglected or does not exist. Accordingly, for three-dimensional cases,  the special space can be neglected or does not exist in most cases. Based on our proposed corollaries, the dimension of the area of interest  can be reduced. Thus the computational complexity of any numeral algorithms can be reduced.
\end{remark}

\section{Extensions and Discussions}

\subsection{Extensions to the Existing of Direct Link}
We consider adding a RIS to enhance the wireless communication in free space. Now, besides the RIS link, the direct link  also exists. The direct channel can be expressed as
\begin{equation}
\mathbf{h}_{TR}=a_{TR}\left[ e^{j\frac{2\pi d_{TR,1}}{\lambda}},  e^{j\frac{2\pi d_{TR,2}}{\lambda}}, \cdots, e^{j\frac{2\pi d_{TR,N}}{\lambda}} \right],
\end{equation}
with
\begin{equation}
a_{TR}=\frac{\sqrt{G_t G_r \lambda^2}}{4\pi}d_{TR}^{-1},
\end{equation}
where $d_{TR,p}$ is the distance from  antenna $p$ to the receiver, $p=1,...,N$. In this scenario, the received signal is a sum from two pathes, which is given by
\begin{equation}\label{powertwopath}
P_{r}=|(\mathbf{h}_{IR}^{H}\bm{\Theta}\mathbf{H}_{TI}^{H}+\mathbf{h}_{TR}^{H})\mathbf{v}|^{2}.
\end{equation}
The joint optimal Problem P1 is unchanged except replacing the objective function with $P_{r}^{'}$.  For this scenario,  we also propose a closed-form phase shifts as follows.
\begin{subequations}
\begin{align}\label{opphasetwopath}
\varphi_{q}^{\star}&\triangleq\frac{\pi}{2}(\frac{O}{|O|}-1)-2\pi\frac{d_{TI}+d_{IR}-d_{TR}+\Delta d^{T}_{I,q}+\Delta d^{R}_{I,q}}{\lambda} \nonumber\\
&=\frac{\pi}{2}(\frac{O}{|O|}-1)-2\pi\frac{1}{\lambda}(d_{TI}+d_{IR}-d_{TR})\nonumber\\
&~~+2\pi\frac{1}{\lambda}\big((\sin\theta_{t}\cos\varphi_{t}+\sin\theta_{r}\cos\varphi_{r})(m_{q}-\frac{M_{I}+1}{2})d_{x}  \nonumber\\
&~~~~~~~~~+(\sin\theta_{t}\sin\varphi_{t}+\sin\theta_{r}\sin\varphi_{r})(n_{q}-\frac{N_{I}+1}{2})d_{y}\big),
\end{align}
\begin{align}
\theta_{q}^{\star}=e^{j\varphi_{q}^{\star}},
\end{align}
\end{subequations}
where
\begin{align}\label{expressionO}
O=\frac{sinc(\frac{N\Delta d_{T}(\cos\mu_{TI}-\cos\mu_{TR})\pi}{\lambda})}{sinc(\frac{\Delta d_{T}(\cos\mu_{TI}-\cos\mu_{TR})\pi}{\lambda})}.
\end{align}
\begin{proof}
See Appendix C.
\end{proof}

Interestingly, compared to (\ref{opphasemain}), only the  term  $\frac{\pi}{2}(\frac{O}{|O|}-1)-2\pi\frac{1}{\lambda}(d_{TI}+d_{IR}-d_{TR})$ is polymeric in (\ref{opphasetwopath}). It' not hard to obtain the corresponding bemforming by using MRT, thus omitted here. With the solutions, the  received power is expressed as
\begin{align}\label{postiontopower1}
P_{r}=NL^{2}a_{TIR}^2P_{t}+Na_{TR}^2P_{t}+2NLa_{TR}a_{TIR}OP_{t}.
\end{align}

The derived results about the orientation and position of the RIS can also be extended to this scenario naturally.  It is obvious that the optimal orientation of RIS is in accordance with the counterpart in Section III-B. The aforementioned conclusions of the optimal position of RIS  has to be adjusted slightly.  The analytic process is similar to that shown in  Section IV. Differently, the plane S  fitting in this scenario can't be arbitrary anymore, and becomes a special plane where the ULA of the transmitter  is perpendicular to it. While analysing in the same manner as the derivation in Appendix B, it is found that only the results under fixed  $d_{TI}$  is available due to the the existing of $O$. Because when $d_{TI}$ is fixed,  $\mu_{TI}$ is fixed, resulting in a certain value of  $O$. By this way, the $F_{object}$ is consistent with the objective function of Problem $P2$.   However, the similar results don't exist anymore when $d_{IR}$ is fixed. Therefore, the theorems and corollaries in Section IV also hold by replacing  area $D~ (d_{TI}\leq d_{TR})\cap(d_{IR}\leq d_{TR}) $ with area $D1~(d_{TI}\leq d_{TR})$.  Besides,  for three-dimensional cases, a now space needs to be added  where the optimal position may lies in. That is the feasible part on the plane consisting of the ULA and line l.

\begin{table*}
\centering
\caption{The Physical and Electromagnetic Parameters of RIS-assisted Wireless Communication}
\begin{tabular}{|c|}
\hline
$P_{t}=0~ dBm, \lambda=0.0286~ m, G=9.03~ dB,   d_{x}=d_{y}=0.01~ m, L=100*100, \Gamma=1, k (in (\ref{DeF}))=3$\\
\hline
$ G_{t}=G_{r}=21~dB,  N=16,  \Delta d=\frac{\lambda}{2}$ \\
\hline
\end{tabular}
 \label{tabparra}
\end{table*}

\subsection{Extension to UPA  Cases}
We have assumed the transmit antenna  to be a ULA in system model. Actually, our work  can be transplanted to  uniform planar array  (UPA) case seamlessly as long as  the far-field condition is satisfied.   The corresponding adjustment is to reformulate the distance approximations  at the transmitter.  In more details, we  reformulate  $\Delta d^{I}_{T,p}$ in (\ref{farappro}) as
\begin{align}
\Delta d^{I}_{T,p}&=-\sin\theta^{T}\cos\varphi^{T}(m^{T}_{p}-\frac{M_{T}+1}{2})d^{T}_{x} \nonumber\\
&~~~-\sin\theta^{T}\sin\varphi^{T}(n^{T}_{p}-\frac{N_{T}+1}{2})d^{T}_{y},
\end{align}
where $M_{T}$ denotes the columns of the UPA and $m^{T}_{p}$ denotes the index number of columns of antenna $p$. $N_{T}$ denotes the rows of the UPA and $n^{T}_{p}$ denotes  the index number of rows of antenna $p$. $d^{T}_{x}\times d^{T}_{y}$ denotes the unit two-dimensional interval in the UPA. $\theta^{T}$ and $\varphi^{T}$ account for the  elevation angle and the azimuth angle from the transmit antenna to the RIS at the transmitter side, respectively. Note that it is similar to the approximations at the RIS since they are both planar array. With the adjustment, our analysis is unchanged for UPA cases.

\subsection{General Beamforming and Phase Shifts}
A main limitation of the aforementioned results for the joint optimization is that they are derived under far-field operation. In far-field operation,   the channel gain is obtained by \emph{far-field amplitude  approximation}  and  \emph{far-field phase approximation}.  These approximations may  be unreliable under near-field operation, when RIS is close to the transmitter or to the receiver or RIS has a large size. In order to extend  application scenarios, we develop the proposed beamforming and phase shifts to both near-field and far-field cases and propose SVD-based solutions. According to (\ref{Pv}),  the received power with optimal $\mathbf{v}$ is
\begin{align}\label{upperboundpower}
P_{r}&=\left\|\bm{\theta}^{T}\mathbf{H}_{TIR}^{H}\right\|^2P_{t}=\left\|\bm{\theta}^{T}\mathbf{U}\Sigma \mathbf{V}^{H}\right\|^2P_{t} \nonumber\\
&\overset{a}{\leq}L\sigma_{max}^{2}(\mathbf{H}_{TIR})\left\|\mathbf{v}_{1}\right\|=L\sigma_{max}^{2}(\mathbf{H}_{TIR}).
\end{align}
where $\mathbf{U}\Sigma \mathbf{V}^{H}$ is the SVD version of $\mathbf{H}_{TIR}^{H}$ and $\mathbf{v}_{1}$ is the first column of $\mathbf{V}$. For $(a)$, the equality holds when $\bm{\theta}=\sqrt{L}\mathbf{u}^{\ast}_{1}$, where $\mathbf{u}_{1}$ represents the first column of $\mathbf{U}$.  However, due to the unit-modulus restriction of phase shifts, we perform a  projection process to obtain a feasible solution. Our proposed phase shifts is given by
\begin{align}\label{generalshifts}
\bm{\theta}^{\star}=\angle{\mathbf{u}^{\ast}_{1}}.
\end{align}
Then, the corresponding beamforming is expressed as
\begin{align}\label{generalbemforming}
\mathbf{v}^{\star}=\frac{\sqrt{P_{t}}}{\sqrt{N}}\frac{(\angle{\mathbf{u}^{H}_{1}}\mathbf{H}_{TIR}^{H})}{\left\|\angle{{\mathbf{u}^{H}_{1}}}\mathbf{H}_{TIR}^{H}\right\|}.
\end{align}

\begin{remark}
As shown in Section III, the channels can be seen as a rank-one unit-modulus matrix in far-filed cases. When applying   (\ref{generalshifts}), the proposed closed-form solutions  for  far-field cases can be derived equivalently. The phase shifts in  (\ref{generalshifts}) is actually a general form of (\ref{opphase}). However, as the distances between communication parties become nearer or the RIS becomes larger, the inequality in (\ref{generalshifts}) is more relax, degrading the effectiveness. Simulation results reveal that  the proposed SVD-based solutions are effective  in  far-field operation and relative near-field operation (the distance is not close enough), but not in absolute near-field operation (the distance is very close).
\end{remark}

\begin{figure}
  \centering
  \includegraphics[width=0.5\textwidth]{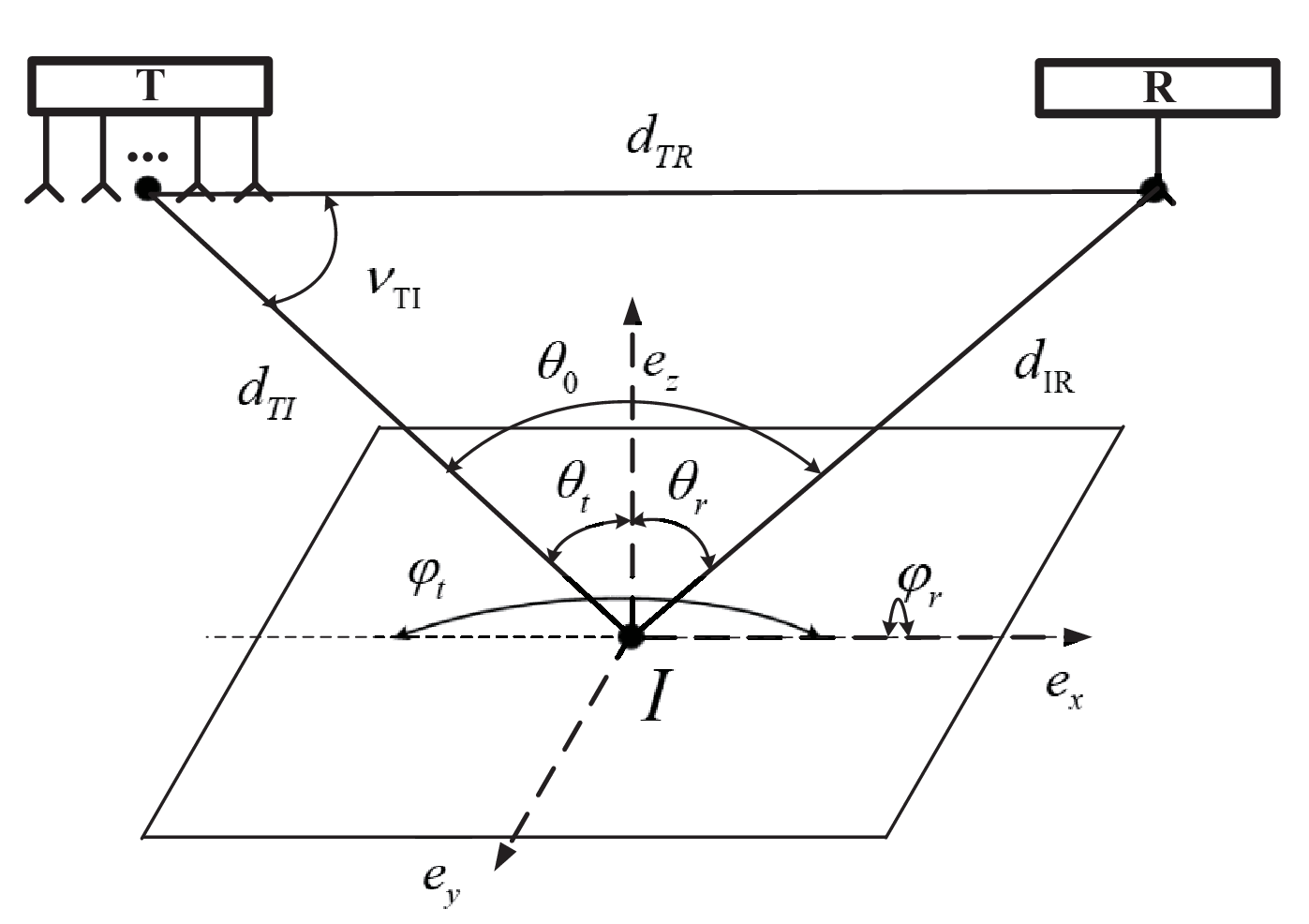}\\
  \caption{Supplementary of the geometric setup in Section VI-A. The line passing T and R and  $e_{x}$  are coplanar, the plane contains T,R and I is  perpendicular to $e_{y}$. The RIS is on the optimal orientation, which indicates that $\theta_{t}=\theta_{r}=\frac{\theta_{0}}{2}$. Besides, $\theta_{t}=\pi,~\theta_{r}=0$.}\label{setup1}
\end{figure}

\section{Simulations and Discussions}

Numerical simulations about RIS-assisted wireless communication in free space have been conducted.  This section is further divided
into three subsections. Subsection VI-A investigates the effectiveness of our proposed beamforming and phase shifts, compared with the received power upper bound as a performance benchmark.  Subsection VI-B studies the received power with the position of RIS on a plane under direct link either blocked or existed. Then we show the advantages of RIS to assist in free-space mmWave/THz communication.   Eventually, we verify the robustness of our proposed beamforming and phase shifts in Subsection VI-C. In default, the physical and electromagnetic parameters are listed in Table.\ref{tabparra}.

\subsection{The Effectiveness of the Proposed Beamforming and Phase Shifts}

We have proposed two approaches for beamforming and phase shifts, one is the closed-form solutions ((\ref{opphase}) and (\ref{opbeaming})), the other is SVD-based solutions ((\ref{generalshifts}) and (\ref{generalbemforming})). The effectiveness of them are demonstrated in Fig.~\ref{Powerscomparasion}. The geometric setups are illustrated in Figure.~\ref{setup1}. Moreover, let $d_{TR}=d_{TI}=d_{IR}=d$. Via scaling the value of $d$, both near-field and far-field operations can be included.   The transmit/receive antennas are omnidirectional here, thus $G_{t}=G_{r}=0~dB$. As seen, when the size of the reflective element is small as $d_{x}=d_{y}=0.01 m$,  both the two proposed approaches achieve the limit performance, equal to the upper bound (\ref{upperboundpower}). When the size of the reflective element is large as  $d_{x}=d_{y}=0.06 m$, the proposed SVD-based method  exceeds the proposed closed-form method, especially at short distances.   But the computational complexity of proposed SVD-based method is higher. It is also found that the proposed SVD-based solutions  still approach the upper bound (\ref{upperboundpower}). Note that, the distances shown in Fig.~\ref{Powerscomparasion}  are already in the near field of the RIS ( see \cite{9206044}). Therefore, our proposed approaches are effective in far-field operation and even in relative near-field operation (the distance is not close enough). For the absolute near-field operation (the distance is very close), the traditional way of calculating the total gain of reflective elements  may not be accuracy  (see \cite{9184098}), therefore we don't consider here.

\begin{figure*}
   \subfloat[\label{sub1}]{%
       \includegraphics[width=0.5\textwidth]{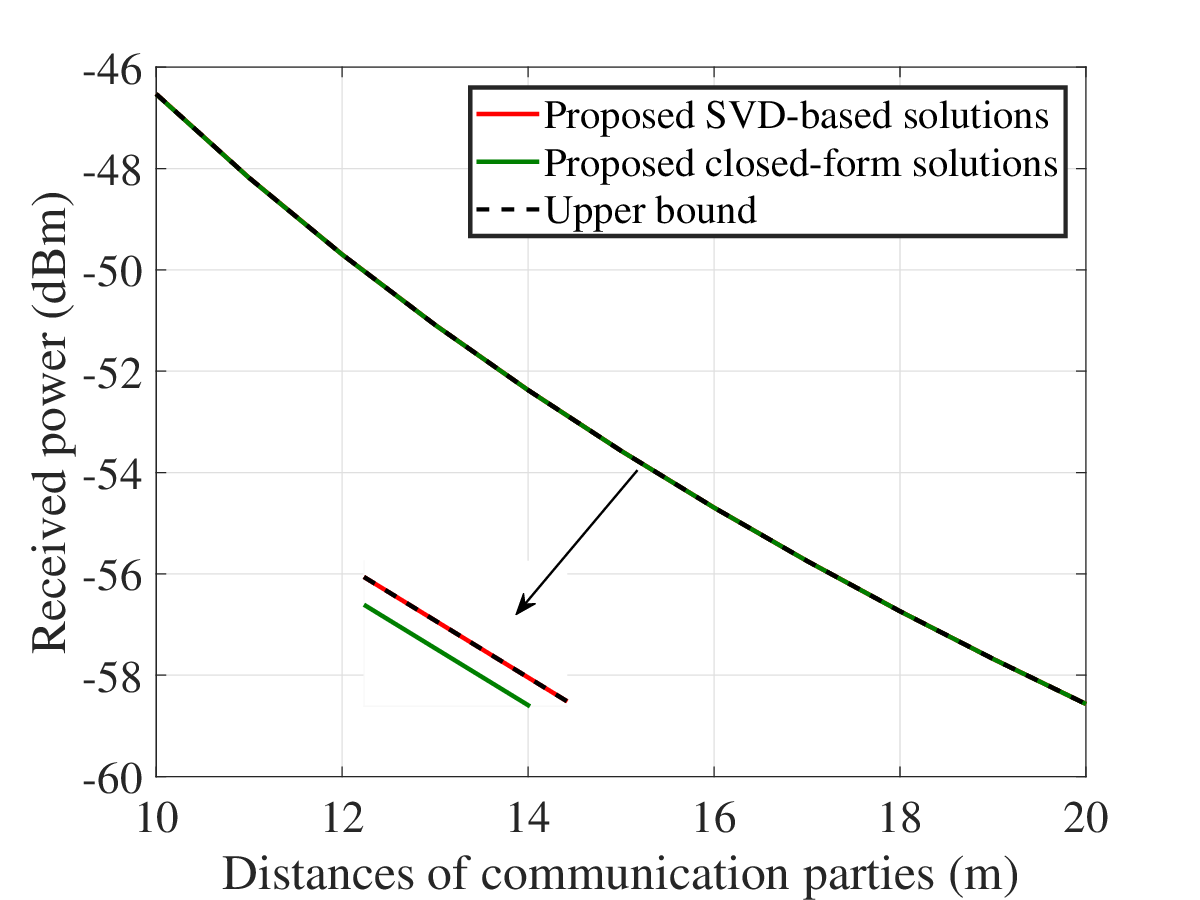}
     }
     \hfill
     \subfloat[\label{sub2}]{%
       \includegraphics[width=0.5\textwidth]{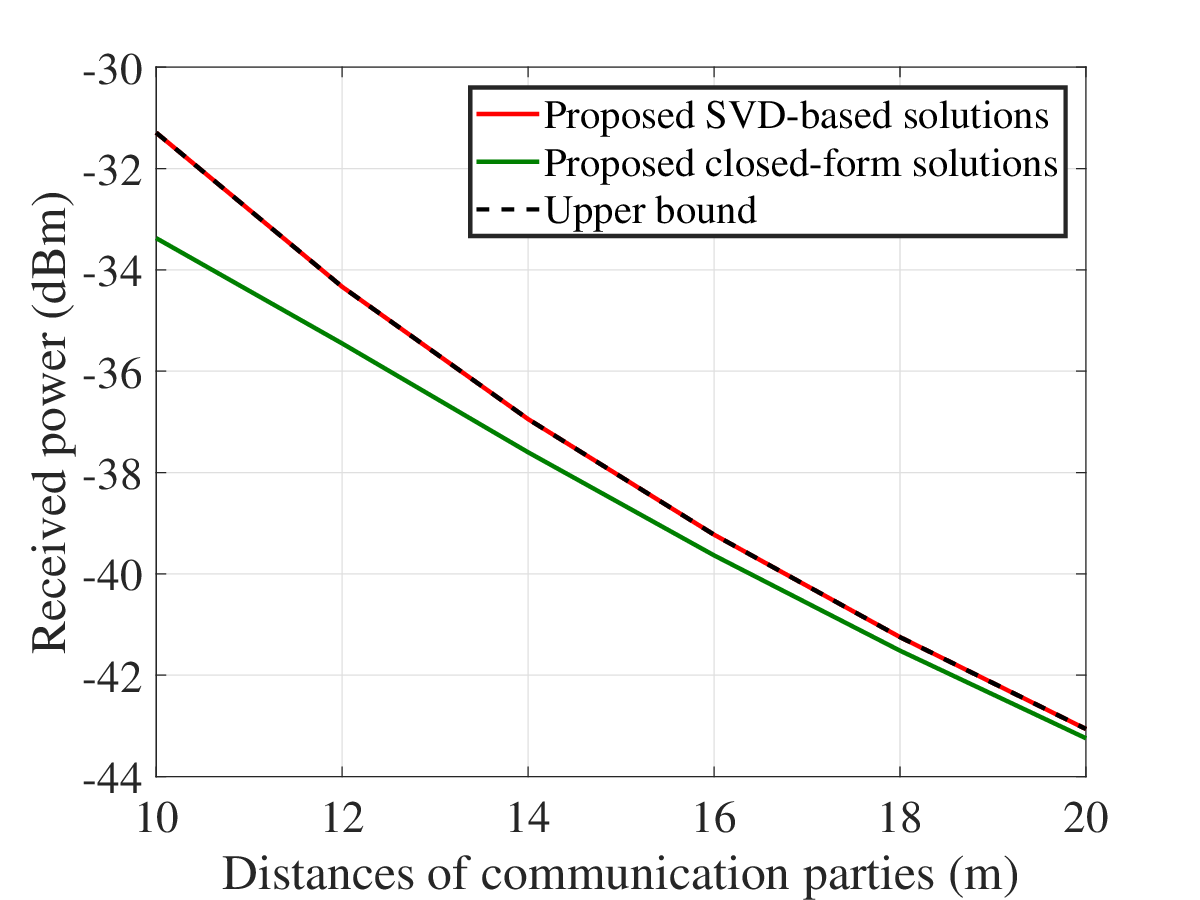}
     }
     \caption{The received power  versus the distances of communication parties. (a) $d_{x}=d_{y}=0.01~m$. (b) $d_{x}=d_{y}=0.06~m$.}
     \label{Powerscomparasion}
\end{figure*}

\subsection{The Performance of  RIS}
\begin{figure*}
   \subfloat[ \label{h80}.]{%
       \includegraphics[width=0.5\textwidth]{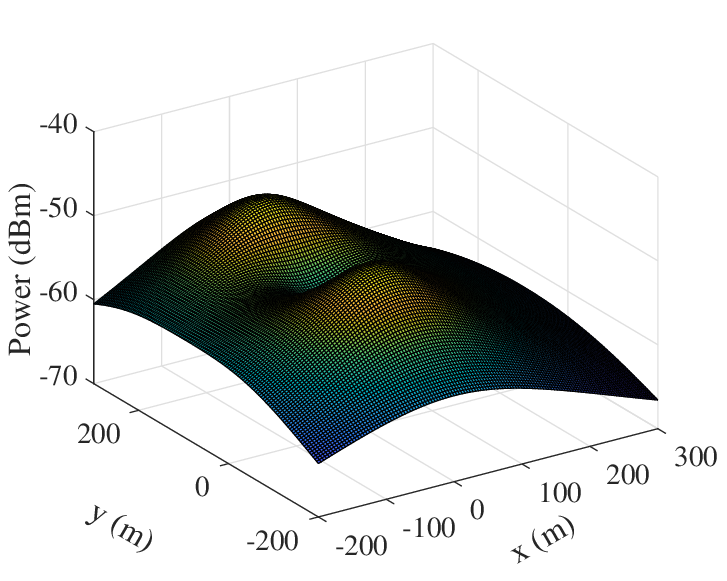}
     }
     \hfill
     \subfloat[\label{h40}.]{%
       \includegraphics[width=0.5\textwidth]{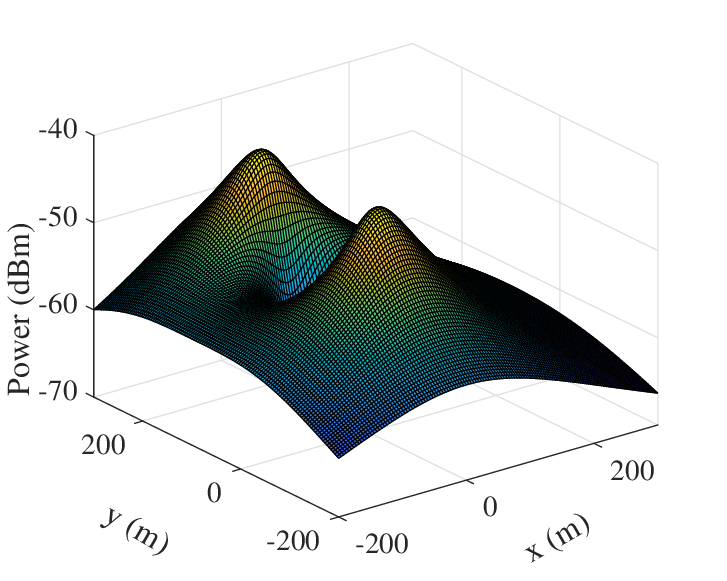}
     }
     \caption{The received power versus the position of the RIS on plane S with only RIS link. (a) $h=80~m$. (b) $h=40~m$. }
     \label{SiPlaneOnePath}
\end{figure*}
\begin{figure*}
   \subfloat[ \label{h80}.]{%
       \includegraphics[width=0.5\textwidth]{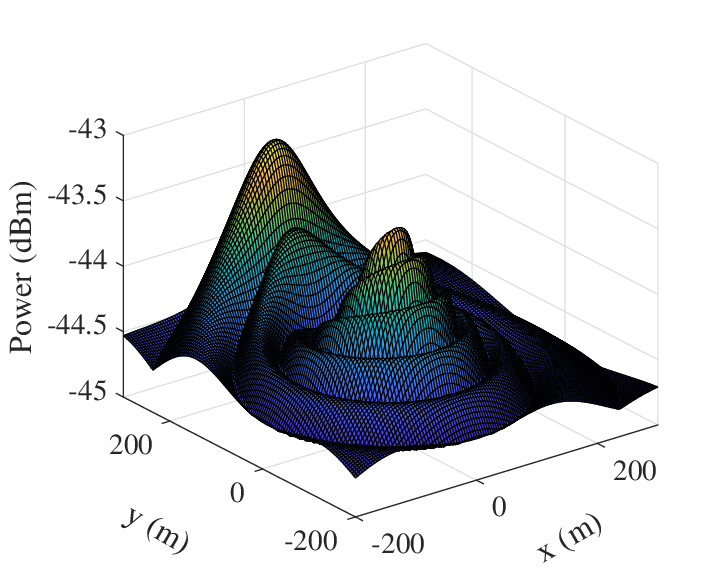}
     }
     \hfill
     \subfloat[\label{h40}.]{%
       \includegraphics[width=0.5\textwidth]{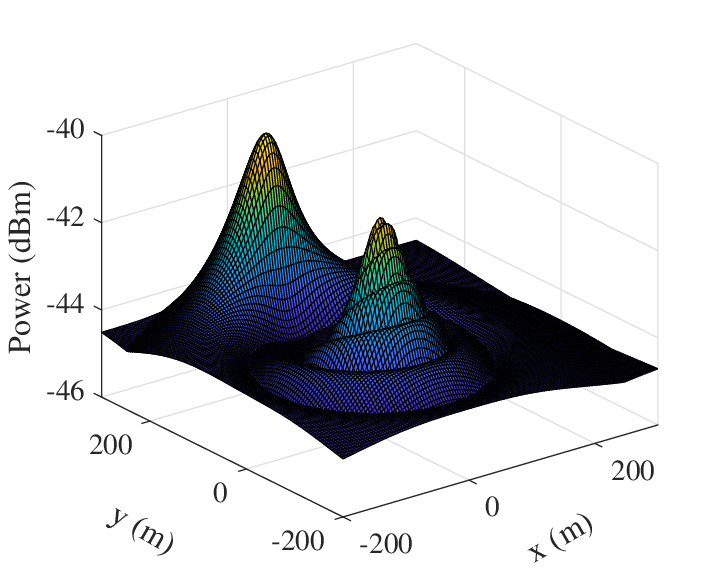}
     }
     \caption{The received power versus the position of the RIS on plane S with both RIS link and direct link. (a) $h=80~m$. (b) $h=40~m$. }
     \label{SiPlaneTwoPath}
\end{figure*}

We investigate the performance of RIS by placing the reference point of RIS on a two-dimensional plane S. For a fixed RIS, the optimal strategy, including the proposed beamforming and phase shifts, the optimal orientation of RIS, is applied\footnote{It is worth mention that with the optimal orientation, the RIS and plane S are not coplanar.}.   The geometric setup can be referred in Fig.~\ref{planeS}.  Moreover, we set $d_{TR}=200~m$ and the direction of ULA is  perpendicular to  plane S (for convenience of simulation but not necessary). The coordinate origins of plane S is selected as $T^{'}$. The  positive direction of X-axis is selected as $l_{\overrightarrow{T'R'}}$ and the Y-axis is determined  correspondingly on plane S. Without loss of generality, we let $h1=h2=h$.

Fig.~\ref{SiPlaneOnePath} reveals the received power versus the position of RIS  on plane S via traversal grid  when the direct link is blocked. As seen the optimal solution is near $T^{'}$ or $R^{'}$  and on line l, which is consistent  with our analysis. Similarly, we investigate the received power for existing  the direct link  in   Fig.~\ref{SiPlaneTwoPath}. As seen, there are many ripples due to the balance of the RIS link and direct link to maximize the received power (\ref{postiontopower1}). In more details, they are determined by the antennas structure and DOAs of the RIS/receiver at the transmitter, which is given by (\ref{expressionO}).  We also find that the amount of the ripples is half of the number of antennas at the transmitter. It is observed that  the optimal solution is close to $R^{'}$, and on line l, which is consistent with  with our analysis. Moreover, the power at $R^{'}$ is obvious larger than  that at $T^{'}$ due to the existence of direct link.

Fig.~\ref{wavelength} illustrates the optimal received power versus the  carrier wavelength.   Moreover, we let $h = 80~m$ and fix the RIS at $R^{'}$ on plane S. Note that the  RIS in this figure satisfies the proposed \emph{anti-decay designing  principle} where $\frac{d_{x}}{ \lambda}=\frac{d_{y}}{ \lambda}=\frac{1}{3}$. It is  observed that, as the wavelength decreases, the optimal received power from direct link decays largely. But  with  the help of RIS link, the  performance degradation is alleviated. Therefore, the RIS has absolute advantages to assist in free-space mmWave/THz communication.

\subsection{The Robustness of Our Proposed Beamforming and Phase shifts}
In practice, the obtained position of RIS may be obtained imperfectly due to the error of measurement or the mobility of the RIS. Therefore we demonstrates the robustness of our proposed solutions of beamforming and phase shifts in Fig.~\ref{robustness}. The geometric setup is same to the last subsection, and the direct link is blocked. Assuming the known position of the RIS is $T^{'}$ (origin). According to it,  the closed-form beamforming and phase shifts are derived, named the estimated solutions. Fig.~\ref{robustness} illustrates the normalized power deviation using estimated solutions  versus the practical position of  the RIS on plane S. The normalized power deviation is used to quantity the deviation of the received power, which is given by
\begin{align}
\hat{P}_{r}=\frac{\|\tilde{P}_{r}-\dot{P}_{r}\|}{\max\{\tilde{P}_{r},\dot{P}_{r}\}},
\end{align}
where $\tilde{P}_{r}$ denotes the received power using estimated solutions and $\dot{P}_{r}$ denotes the ideal optimal received power. It is found that the area for $\hat{P}_{r}<0.1$ is large than $5~m\times 5~m$. Therefore, our proposed solutions are strongly robust to  the position perturbation of the RIS. Note that the performance of the proposed SVD-based solutions is consistent with the proposed closed-form solutions in the far-field operation, thus omitted here.

\begin{figure}
  \centering
  \includegraphics[width=0.5\textwidth]{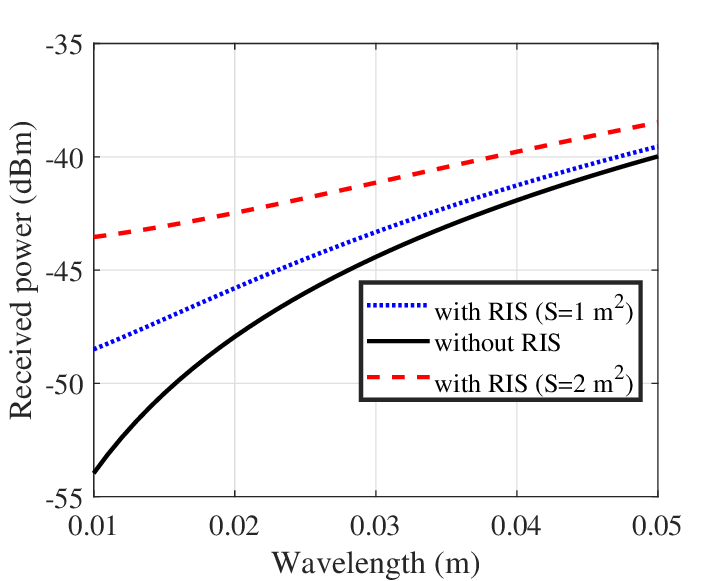}\\
  \caption{The optimal received power versus the carrier wave length.}\label{wavelength}
\end{figure}

\begin{figure*}
  \subfloat[ \label{h80}.]{%
       \includegraphics[width=0.5\textwidth]{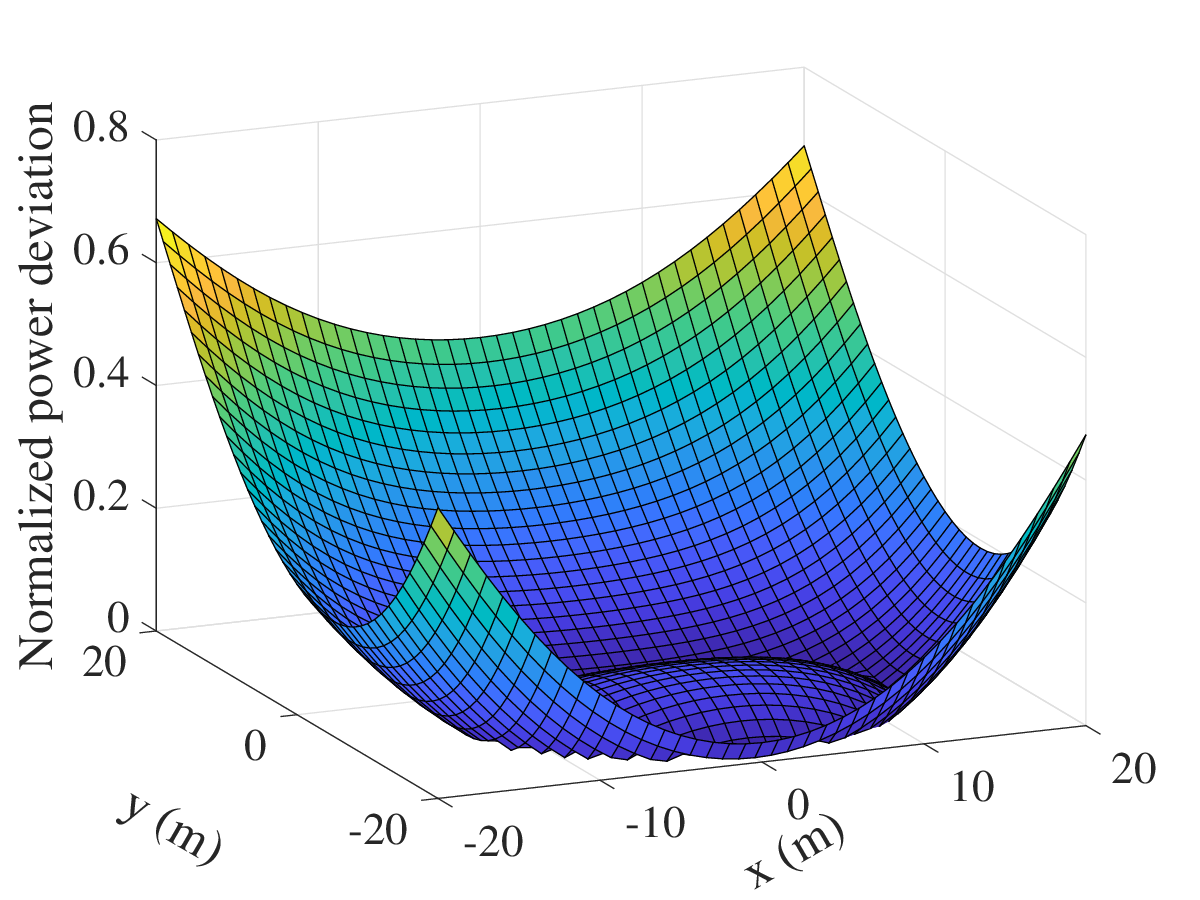}
     }
     \hfill
     \subfloat[\label{h40}.]{%
       \includegraphics[width=0.5\textwidth]{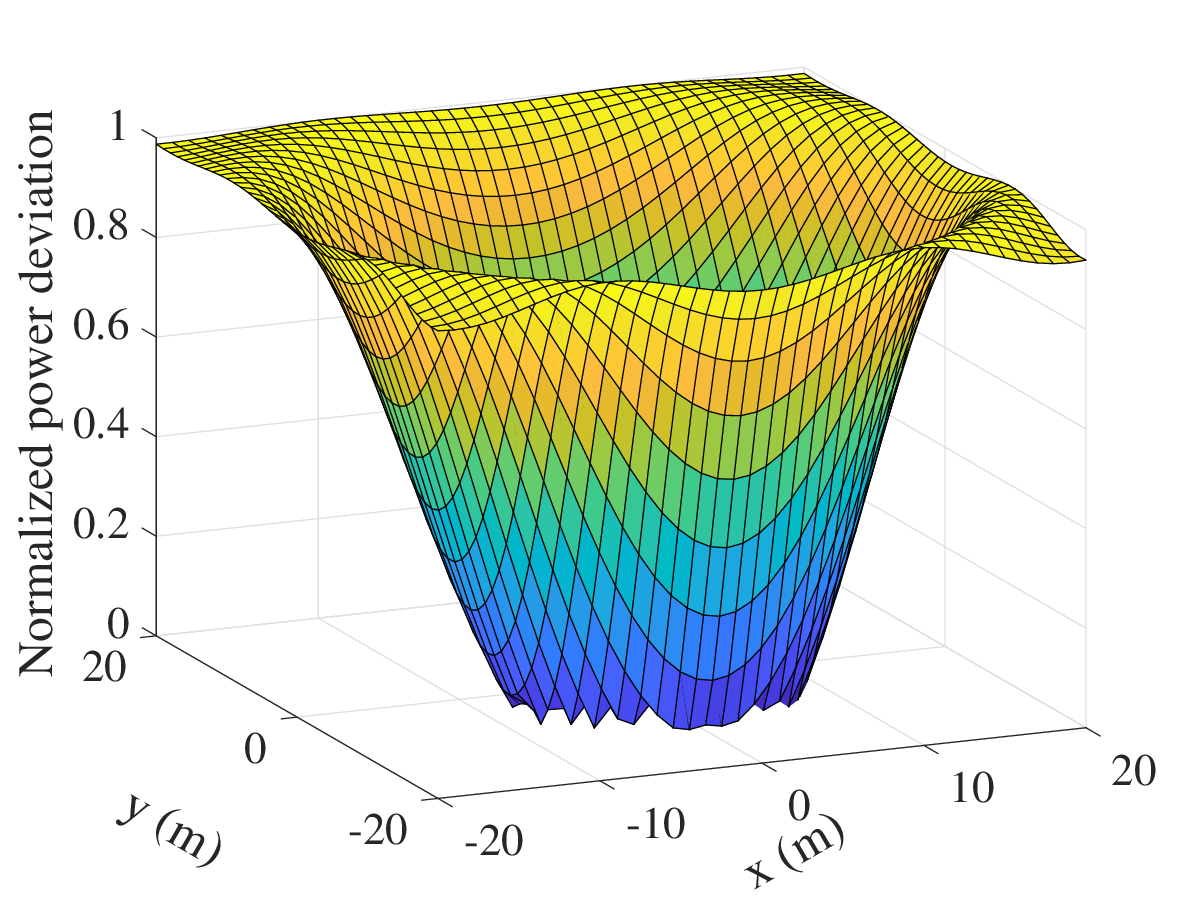}
     }
     \caption{The normalized received power deviation  versus the  position perturbation  of  RIS on plane S. (a) $h=80~m$. (b) $h=40~m$. }
     \label{robustness}
\end{figure*}

\section{Conclusion}
In this paper, comprehensive optimization of incorporating a RIS to MISO wireless communication  in free space has been considered from electromagnetic and physical  perspectives. The closed-form solutions of transmitter's beamforming and phase shifts have been proposed and extended. Considering the general power radiation pattern, we have proved  that the optimal orientation  of  the RIS is just to satisfy specular reflection. United with the above contributions, the position optimizing problem  of placing a RIS  has been studied. For most three-dimensional space to place the RIS, a substantial dimensionality reduction theory was provided. In simulation part, the proposed closed-form solutions of beamforming and phase shifts approach the power upper bound. Besides, the robustness in terms of position perturbation is verified. The simulation results indicate that adding a RIS is remarkable in mmWave/THz communication.
\appendices

\section{Proof of Theorem. \ref{th1} ($k=0$)}
For any fixed $\nu_{TI}$, the orbit of feasible position is a circle $C(\nu_{TI})$ on plane S.  If $\nu_{TI}$ is fixed, the value of $d_{TI}$ is also certain. $F_{object}(d_{IR})$ decreases from $p_{1}(\nu_{TI})$ to $p_{2}(\mu_{TI})$ along $C(\nu_{TI})$. $F_{object}(d_{IR})$ is a decreasing function to $d_{IR}$, the max value of $F_{object}$ arrives  at point $p_{1}(\nu_{TI})$. Therefore, for an arbitrarily  position on the plane S, there must exists a position on the half-line $l_{\overrightarrow{T'R'}}$ with the same $\nu_{TI}$, at which the value of the objective function is equal or bigger. As a conclusion,  the optimal position is must on the half-line $l_{\overrightarrow{T'R'}}$. {With the same manner, when starting from fixing $\nu_{RI}$, we obtain a conclusion that  the optimal position is must on the half of the half-line $l_{\overrightarrow{R'T'}}$. Therefore, on plane S, the optimal position of  RIS is must on  line segment $l_{R'T'}$.

\section{Proof of Theorem. \ref{th1} ($k>0$)}
Necessarily, we'll exploit the quasi-convex property of $F_{object}$. For convenient expression, we simplify $d_{IR}$ as $x$ and denote $\sqrt{F_{object}(x)}$ as $F$. To maximize the value of $F_{object}$ is equal to maximize the value of $F$. Then the objective function $F(x)$ can be written as
\begin{align}\label{F_1}
F(x)=x^{-1}(\underbrace{ax^{-1}+bx+\frac{1}{2}}_{f(x)})^{\frac{k}{2}}~~~ x>0.
\end{align}
Wherein, the constant $a=\frac{d_{TI}^2-d_{TR}^2}{4d_{TI}}, b=\frac{1}{4d_{TI}}>0,k>0$ and $0<f(x)<1$, deducing from (\ref{OpFF}).

The derivation of $F(x)$, denoted as $F^{'}(x)$ is given by
\begin{align}\label{D_F_1}
F^{'}(x)=-f(x)^{\frac{k}{2}}x^{-2}+\frac{k}{2}x^{-1}f(x)^{(\frac{k}{2}-1)}f^{'}(x),
\end{align}
where,
\begin{align}\label{order-one}
f^{'}(x)=(-ax^{-2}+b).
\end{align}

\begin{IEEEproof}[proof ($0<k\leq2$)]
\begin{align}
F^{'}(x)&=[f(x)^{(\frac{k}{2}-1)}x^{-2}]\underbrace{(-f(x)+\frac{k}{2}xf^{'}(x))}_{g(x)}.
\end{align}
Because $0<f(x)$, $g(x)$ determines whether $F^{'}(x)$ is positive or negative.
\begin{align}
g(x)&=-(ax^{-1}+bx+\frac{1}{2})+\frac{k}{2}(-ax^{-1}+bx) \nonumber\\
&=-(\frac{k}{2}+1)ax^{-1}+(\frac{k}{2}-1)bx-\frac{1}{2}<0.
\end{align}
So, $F(x)$ is a quasi-convex (quasilinear) function for $a>0,~0<k\leq2$.
\end{IEEEproof}

\begin{IEEEproof}[proof ($k\geq2$)]
The condition for $F^{'}(x)=0$ is
\begin{align}\label{eq1}
f(x)=\frac{k}{2}xf^{'}(x).
\end{align}
If the solution doesn't exist, then $F(x)$ is a quasi-convex (quasi-linear) function. If it exists, since $f(x)>0$, $(\ref{eq1})$ implies $f^{'}(x)>0$.  The second derivative of $F(x)$ is given by
\begin{align}
F^{''}(x)&=(x^{-3}f(x)^{\frac{k}{2}-2}) \nonumber\\
&~~~(-\frac{k}{2}xf(x)f^{'}(x)+2f(x)^{2}-\frac{k}{2}xf(x)f^{'}(x) \nonumber\\ &~~~~+\frac{k}{2}(\frac{k}{2}-1)x^2f^{'}(x)^2+\frac{k}{2}x^2f(x)\underbrace{2ax^{-3}}_{f^{''}(x)})  \nonumber\\
&\triangleq (x^{-3}f(x)^{\frac{k}{2}-2})h(x),
\end{align}
Due to $f(x)>0$, whether $F^{''}(x)$ is negative or positive is determined by $h(x)$. The expression of $h(x)$ is further expressed as
\begin{align}
h(x)&=f(x)\underbrace{[-\frac{k}{2}xf^{'}(x)+f(x)]}_{Q}+f(x)^2\nonumber\\
&~~-\frac{k}{2}xf^{'}(x)\underbrace{[f(x)-\frac{k}{2}xf^{'}(x)]}_{Q} \nonumber\\
&~~-\frac{k}{2}x^2f^{'}(x)^2+\frac{k}{2}x^2f(x)f^{''}(x) \nonumber\\
&\overset{a}{=}f(x)^2-\frac{k}{2}x^2f^{'}(x)^2+\frac{k}{2}x^2f(x)f^{''}(x)  \nonumber\\
&\overset{b}{=}(\frac{k^2}{4}-\frac{k}{2})x^2f^{'}(x)^2+kax^{-1}f(x)>0,
\end{align}
wherein, $(a)$ is due to $Q=0$ when (\ref{eq1}) holds. $(b)$ results from the substitution of (\ref{eq1}) and the expression of $f^{''}(x)$.
Therefore, $F(x)$ is a quasi-convex function for $a>0, ~k>2$.
\end{IEEEproof}
As a conclusion, when $d_{TI}>d_{TR}$, $F_{object}(d_{IR})$ is a quasi-convex function for a fixed $d_{TI}$. Due to the symmetry  of  $d_{TI}$ and $d_{IR}$ in the function $F_{object}$, it also holds  that when  $d_{IR}>d_{TR}$, $F_{object}(d_{TI})$ is a quasi-convex function for a fixed $d_{IR}$.
Let  D1 account for the area where $d_{TI}\leq d_{TR}$ and D2 account for the area where $d_{IR}\leq d_{TR}$. The union set of D1 and D2 is denoted as D.

Eventually, based on the basic property of  quasi-convex function (Section $3.4.2$ in \cite{Boyd2004Convex}), extending Appendix A, we derive a conclusion for $k>0$.  For $k>0$, on the area S-D, the optimal position is must on the line l. Note that it is the line l, not the line segment $l_{R'T'}$ due to the discrepancy of monotonicity and quasi-convex property.

\section{Optimal Phase Shifts for Existing Direct Link}
We already know the MRT is the optimal  method to design beamforming for point to point communication. After applying MRT, the received power can be expressed as
\begin{align}\label{PowerforTheta}
P_{r}&=\left\|a_{TIR}e^{j2\pi\frac{d_{TI}+d_{IR}}{\lambda}}\bm{\theta}^{T}\mathbf{d}\mathbf{b}^{T}+a_{TR}e^{j2\pi\frac{d_{TR}}{\lambda}}\mathbf{e}^{T}\right\|^2P_{t} \nonumber\\
&=\left\|a_{TIR}\underbrace{e^{j2\pi\frac{d_{TI}+d_{IR}-d_{TR}}{\lambda}}\bm{\theta}^{T}\mathbf{d}}_{A'}\mathbf{b}^{H}+a_{TR}\mathbf{e}^{T}\right\|^{2}P_{t},
\end{align}
where,
\begin{subequations}
\begin{align}
\mathbf{e}^{T}=\left[e^{j\frac2\pi\frac{\Delta d^{R}_{T,1}}{\lambda}}, e^{j2\pi\frac{\Delta d^{R}_{T,2}}{\lambda}},\cdots, e^{j2\pi\frac{\Delta d^{R}_{T,N}}{\lambda}},\right],
\end{align}
\begin{align}
\Delta d^{R}_{T,p}=(\frac{(N+1)}{2}-p)\cos{\mu_{TR}}\Delta d_{T}.
\end{align}
\end{subequations}

We equivalently represent  $A'$ as $Me^{jx}$, which is given by
\begin{align}\label{A}
Me^{jx}=\sum_{q=1}^L\underbrace{e^{j(\varphi_{q}+2\pi\frac{(d_{IR}+d_{TI}-d_{TR}+\Delta d^{T}_{I,q}+\Delta d^{R}_{I,q})}{\lambda})}}_{A'_{q}}.
\end{align}
It indicates that,  via changing $\varphi_{q}$, the amplitude $M$ can be an arbitrary value in the feasible set $[0, L]$ and the amplitude $x$ can be an arbitrary value in the feasible set $[0,2\pi]$.

Substituting (\ref{A}) into (\ref{PowerforTheta}) with new formulation of A', we obtain
\begin{align}
P_{r}&=\sum_{p=1}^{N}|a_{TIR}ME_{1}(p)+a_{TR}E_{2}(p)|^2P_{t} \nonumber\\
&=NM^2a_{TIR}^2P_{t}+Na_{TR}^2P_{t} \nonumber\\
&~+2Ma_{TIIR}a_{TR}\underbrace{\sum_{p=1}^{N}\cos{(\angle{E_{1}(p)}-\angle{E_{2}(p)})}}_{E_{3}}P_{t},
\end{align}
where,
\begin{subequations}
\begin{equation}
E_{1}(p)=e^{jx+j2\pi(\frac{(N+1)}{2}-p)\cos{\mu_{TI}}\frac{\Delta d_{T}}{\lambda}},
\end{equation}
\begin{equation}
E_{2}(p)=e^{j2\pi(\frac{(N+1)}{2}-p)\cos{\mu_{TR}}\frac{\Delta d_{T}}{\lambda}},
\end{equation}
\begin{equation}
\begin{aligned}
&~~~~~~~~~~~~~~~~~~~E_{3}=\sum_{p=1}^{N}\cos{(x-K_{p}}), \\
&K_{p}\triangleq2\pi(p-\frac{(N+1)}{2})(\cos{\mu_{TI}}-\cos{\mu_{TR}})\frac{\Delta d_{T}}{\lambda}.
\end{aligned}
\end{equation}
\end{subequations}

We deduce the expression of $E_{3}$ further as follows.
\begin{align}
E_{3} & =\cos(x)\sum_{p=1}^N\cos(K_p)+\sin(x)\sum_{p=1}^N\sin(K_p) \nonumber\\
& \overset{a}{=}\cos(x)\sum_{p=1}^N\cos(K_p)+j\cos(x)\sum_{p=1}^N\sin(K_p) \nonumber\\
& =\cos(x)\sum_{p=1}^N e^{jK_p} \nonumber\\
& \overset{b}{=}N\cos(x)\underbrace{\frac{\mathrm{sinc}(\frac{N \Delta d_{T}(\cos\mu_{TI}-\cos \mu_{TR})\pi}{\lambda})}{\mathrm{sinc}(\frac{\Delta d_{T}(\cos \mu_{TI}-\cos\mu_{TR})\pi}{\lambda})}}_{O},
\end{align}
where $(a)$  is due to $\sum_{p=1}^N\sin(K_p)=0$ and $(b)$ results from $e^{jK_p} (p=1:N)$ is a geometric progression. It is verified that the conditions of maximizing $P_{r}$ are  $M=L$ and $\cos(x)= \frac{O}{|O|}$. The variable $A'$ can achieve this requirement if and only if
\begin{align}
\varphi_{q}^{\star}&\triangleq\frac{\pi}{2}(\frac{O}{|O|}-1)-2\pi\frac{d_{TI}+d_{IR}-d_{TR}+\Delta d^{T}_{I,q}+\Delta d^{R}_{I,q}}{\lambda} \nonumber\\
&=\frac{\pi}{2}(\frac{O}{|O|}-1)-2\pi\frac{1}{\lambda}(d_{TI}+d_{IR}-d_{TR})\nonumber\\
&~+2\pi\frac{1}{\lambda}\big((\sin\theta_{t}\cos\varphi_{t}+\sin\theta_{r}\cos\varphi_{r})(m_{q}-\frac{M_{I}+1}{2})d_{x}  \nonumber\\
&~~~~~~~~+(\sin\theta_{t}\sin\varphi_{t}+\sin\theta_{r}\sin\varphi_{r})(n_{q}-\frac{N_{I}+1}{2})d_{y}\big).
\end{align}

\ifCLASSOPTIONcaptionsoff
  \newpage
\fi

\bibliographystyle{IEEEtran}
\bibliography{IEEEfull,cite}

\end{document}